\documentclass[lettersize,journal,10pt,twoside]{IEEEtran}
\IEEEoverridecommandlockouts
\usepackage{cite}
\usepackage{amsmath,amssymb,amsfonts}
\usepackage{algorithmic}
\usepackage{bm}
\usepackage{graphicx}
\usepackage{textcomp}
\usepackage{xcolor}
\usepackage{enumitem}
\graphicspath{{./Figure/}}
\usepackage[ruled,linesnumbered]{algorithm2e}
\usepackage{color}
\usepackage{epstopdf}
\usepackage{upgreek}
\makeatletter
\usepackage{array}
\usepackage{stfloats}
\usepackage{url}
\usepackage{verbatim}
\usepackage{multirow}
\usepackage{booktabs}
\usepackage{makecell}
\usepackage{subcaption}
\usepackage{caption}
\usepackage{cuted} 
\usepackage{float}
\usepackage{bm}
\usepackage{upgreek}
\usepackage{mathrsfs}
\usepackage{type1cm}
\usepackage[T1]{fontenc} 
\renewcommand{\maketag@@@}[1]{\hbox{\m@th\normalsize\normalfont#1}}
\usepackage{balance}
\usepackage{soul,color}
\usepackage[normalem]{ulem}
\newcommand\redst{\bgroup\markoverwith{\textcolor{red}{\rule[0.2ex]{2pt}{1pt}}}\ULon}

\makeatother
\usepackage[letterpaper,bottom=1.01in,top=0.751in,left=0.631in,right=0.631in]{geometry}
\newtheorem{Theorem}{Theorem}

\title
{Design of Uplink ISAC Systems with Cooperative Sensing: Power Control and Receive Beamforming}
\author{Ling He, \IEEEmembership{Graduate Student Member},
Vaibhav Kumar,  \IEEEmembership{Member, IEEE},
\\
Roberto Bomfin, \IEEEmembership{Member, IEEE},
Yingyang Chen, \IEEEmembership{Senior Member, IEEE},
\\
Miaowen Wen, \IEEEmembership{Senior Member, IEEE},
and 
Marwa Chafii, \IEEEmembership{Senior Member, IEEE},
\thanks{The work of Ling He was supported by the China Scholarship Council (CSC) under Grant No. 202406150151, during her visiting Ph.D. study at New York University Abu Dhabi (NYUAD), UAE.
The work of Vaibhav Kumar and Marwa Chafii was supported by the Center for Cyber Security through NYUAD Research Institute under Award G1104. The work of Marwa Chafii was also supported by Tamkeen under the Research Institute NYUAD grant CG017.
The work of Yingyang~Chen was supported by the Guangdong Basic and Applied Basic Research Foundation under Grant 2024B1515020002.
The work of Miaowen~Wen was supported by the Fundamental Research Funds for the Central Universities under Grant 2024ZYGXZR076.}
\thanks{Ling He is with the School of Electronic and Information Engineering, South China University of Technology, Guangzhou 510640, China, and also with the Wireless Research Lab, Engineering Division, New York University (NYU) Abu Dhabi, UAE (e-mail: eelinghe@mail.scut.edu.cn).}
\thanks{Vaibhav Kumar, Roberto Bomfin, and Marwa Chafii are with the Wireless Research Lab, Engineering Division, NYU Abu Dhabi, UAE. Marwa Chafii is also with NYU WIRELESS, NYU Tandon School of Engineering, New York, USA (e-mail: vaibhav.kumar@ieee.org; 
roberto.bomfin@nyu.edu;
marwa.chafii@nyu.edu).}
\thanks{Yingyang Chen is with the College of Information Science and Technology, Jinan University, Guangzhou 510632, China (e-mail: chenyy@jnu.edu.cn).}
\thanks{Miaowen Wen is with the School of Electronic and Information Engineering and the Guangdong Provincial Key Laboratory of Short-Range Wireless Detection and Communication, South China University of Technology, Guangzhou 510640, China (e-mail: eemwwen@scut.edu.cn).}
}

\begin{document}
\maketitle

\begin{abstract}
Integrated sensing and communication (ISAC) has emerged as a key paradigm for next-generation wireless systems, which allows wireless resources to be used for data transmission and target sensing simultaneously.
In this paper, multi-user collaborative target detection in the uplink ISAC system is investigated.
To incorporate the target sensing functionality, the system relies on the reuse of uplink signals from the communication users.
Specifically, we analyze an uplink multi-user single-input multiple-output (MU-SIMO) communication system with bistatic sensing. 
Using the channel statistics, we formulate the problem of joint optimal pilot and data power allocation to maximize the uplink ergodic sum rate while meeting communication and sensing quality-of-service (QoS) requirements. 
To address this non-convex problem, we propose an alternating optimization (AO)-based iterative framework, where the joint power allocation problem is decomposed into two sub-problems. Specifically, the pilot power allocation is optimized using a penalty dual decomposition (PDD)-based gradient ascent algorithm, while the data power allocation is solved via successive convex approximation (SCA). Once the long-term power allocation is determined, the base station (BS) estimates the instantaneous channels using a minimum mean-squared error (MMSE) estimator. Subsequently, based on the estimated instantaneous channel state information (CSI), the receive beamforming for communication users is optimized via another SCA-based method to maximize the sum rate. Meanwhile, the optimal receive beamforming for the target is obtained in closed-form through eigenvalue decomposition (EVD).
We provide comprehensive simulation results to analyze the performance of the proposed iterative algorithm and to demonstrate its dependence on different design parameters. 
Our results also confirm the superiority of the proposed resource allocation approach over conventional benchmark schemes. 
\end{abstract}


\section{Introduction}
\IEEEPARstart{T}{he} explosive growth of wireless services and the widespread deployment of smart devices are reshaping the architecture of modern wireless networks. 
With the emergence of applications such as autonomous driving, smart cities, industrial automation, and intelligent healthcare~\cite{23-COMST-6G,21-NETW-ISAC,25-OJCOMS-ISAC-RIS}, future wireless systems are expected to not only deliver high-speed and reliable communication services but also offer high-resolution and real-time environmental sensing capabilities~\cite{20-NETW-6G,22-COMST-ISAC,23-IoTJ-ISAC}. 
Traditionally, communication systems and sensing systems have been developed and deployed independently, relying on different hardware, spectrum resources, and signal processing techniques.
This segregated approach results in low spectrum utilization, high system costs, and poor functional integration, which seriously restricts the collaborative development and large-scale expansion of wireless networks. 
Fortunately, with the advancement of wireless communication technologies, sensing and communication tend to be similar in terms of system architectures and signal processing algorithms~\cite{19-TVT-ISAC}, which enables the joint design and operation of communication and sensing functions within the same wireless system.
In this direction, integrated sensing and communication (ISAC)~\cite{22-JSAC-ISAC} has emerged as a groundbreaking technology for next-generation wireless networks to cater to the ever-growing demand for high-quality wireless connectivity and accurate sensing capabilities.
By sharing spectrum, hardware, and signaling strategies, ISAC systems can improve spectral efficiency and reduce hardware duplication compared with separately deployed sensing and communication systems~\cite{21-VTM-ISAC-PMN}.

While the majority of existing ISAC research has focused on downlink transmission, where the base station (BS) actively emits either a dual-function waveform or a superposition of communication and sensing signals over the same spectrum, this conventional monostatic paradigm nevertheless requires dedicated downlink transmission resources in terms of power and waveform design degrees of freedom~\cite{22-JSAC-ISAC,23-TVT-ISAC}. By contrast, uplink ISAC has recently emerged as a promising alternative~\cite{22-IDAC-UL}. In the uplink framework, the BS operates as a bistatic receiver that exploits users’ uplink communication signals for target sensing. This uplink-centric paradigm enables sensing without additional downlink transmit power or signaling overhead, and naturally supports distributed or cooperative sensing scenarios, such as vehicular networks and Internet-of-Things (IoT) applications. Nevertheless, uplink ISAC also introduces unique challenges, including decentralized signal transmission from multiple users, more pronounced multi-user interference, and the need for joint optimization of user transmit power and BS receive beamforming to balance communication and sensing performance. Motivated by these considerations, this paper investigates the design and optimization of multi-user uplink ISAC systems with cooperative sensing.

\subsection{Related Works}
In recent years, ISAC technology has rapidly emerged as a focus of intensive research across both academia and industry. 
The theoretical foundation of ISAC lies in the waveform duality and functional compatibility between communication signals and sensing waveforms~\cite{11-PI-ISAC-Wave}. 
The waveforms widely used in communication systems, such as orthogonal frequency-division multiplexing (OFDM)~\cite{20-TCOM-ISAC}, orthogonal time frequency space (OTFS)~\cite{21-JSTSP-ISAC-OTFS}, and multiple-input multiple-output (MIMO) architectures~\cite{22-COMST-ISAC-MN}, inherently possess time-frequency properties (\emph{e.g.}, orthogonality and spatial sparsity) suitable for high-precision parameter estimation. 
Consequently, ISAC systems offer new possibilities for joint waveform design~\cite{22-TVT-ISAC-RIS}, resource allocation~\cite{23-TWC-ISAC-UAV}, and beamforming strategies~\cite{23-TVT-ISAC}, enabling the system to achieve a flexible trade-off between communication quality and sensing accuracy.

Research in ISAC systems has primarily concentrated on the aforementioned downlink scenarios, considering fundamental trade-offs between communication and sensing performance~\cite{11-PI-ISAC-Wave,18-TSP-ISAC,22-JSAC-ISAC-FD} and joint optimization theories for cross-domain resource coordination~\cite{23-TWC-ISAC,24-TWC-ISAC-Pilot,22-TSP-ISAC-RIS}.
In the early stages of ISAC research, Sturm and Wiesbeck~\cite{11-PI-ISAC-Wave} explored the multiplexing potential of radar waveforms in communications, laying the foundation for ISAC signal design.
This work demonstrated the feasibility of joint waveform reuse and spectrum sharing, providing critical insights for enhancing spectral efficiency and hardware integration.
Building upon this pioneering demonstration of cross-system signal compatibility, subsequent research efforts have evolved toward more flexible and reconfigurable waveform architectures.
For instance, Liu \emph{et al.}~\cite{18-TSP-ISAC} proposed a dual-functional MIMO-OFDM waveform design framework that supports both high-resolution radar sensing and reliable wireless communication.
By dynamically adjusting weighting coefficients, this framework achieved adaptive trade-offs between communication and sensing.
While such frameworks leveraged spatial degrees of freedom for dual-functional operation, the practical deployment in full-duplex scenarios requires careful consideration of self-interference impacts.
To ensure the quality-of-service (QoS) of the ISAC system, the authors in~\cite{22-JSAC-ISAC-FD} designed a comprehensive framework for waveform design and performance analysis in full-duplex ISAC systems. Effective suppression of self-interference was achieved by appropriate power control of radar and communication signals.

However, waveform design alone is insufficient to meet the optimal dual performance requirements of both communication and sensing systems, especially in multi-user scenarios and highly dynamic environments. 
Consequently, an increasing number of studies focus on the optimization of system resources in downlink systems. More effective collaboration and performance enhancement between communication and sensing functions are achieved by jointly optimizing central resource management strategies, such as power allocation and beamforming design.
For instance, An \emph{et al.}~\cite{23-TWC-ISAC} investigated a power control strategy under the target detection probability and communication achievable rate constraints.
On the basis of power control, the authors in~\cite{24-TWC-ISAC-Pilot} further studied the spatial flexibility enabled by beamforming design and proposed a novel protocol for target detection and information transmission. By jointly designing the pilot matrix, training duration, and transmit beamforming, the proposed scheme effectively enhances target sensing performance while maintaining communication efficiency.
As ISAC systems gradually extend to multi-user environments, it becomes necessary to address more complex joint optimization problems. 
For example, Yu \emph{et al.}~\cite{22-TSP-ISAC-RIS} investigated a location-aware beamforming design scheme for multi-user ISAC systems, avoiding high channel estimation overhead by jointly optimizing the active and passive beamforming designs.

In line with the increasing demands on uplink QoS, research has recently extended to multi-user uplink scenarios.
Unlike downlink systems, uplink scenarios involve multiple users simultaneously and independently transmitting signals to the BS. 
This decentralized nature of signal transmission introduces more complex design challenges, particularly in terms of resource management, interference suppression, and coordination between sensing and communication.
To tackle these challenges, recent studies have started exploring the joint optimization frameworks for multi-user uplink ISAC systems~\cite{24-TVT-ISAC,25-TCCN-ISAC-UL}.
For instance, the authors in~\cite{24-TVT-ISAC} investigated the frequency-time
resource allocation scheme to minimize the Cramér–Rao Bound (CRB) of 5G New Radio (NR) signals on the transmitter side.
As for the receiver side, Li \emph{et al.}~\cite{25-TCCN-ISAC-UL} proposed a cooperative sensing scheme based on multiple roadside units, in which the receive beamforming is optimized to minimize the sensing beampattern matching error while satisfying the communication QoS.

\subsection{Contributions and Organization}
Most of the aforementioned studies adopt the BS as the active signal source, thereby overlooking the sensing potential inherent in user-transmitted signals during uplink communication. In practice, uplink pilot and data signals can simultaneously support communication and target sensing. Ignoring this information may constrain sensing performance, particularly in multi-user uplink scenarios where rich spatial and signal diversity is available. Moreover, existing uplink ISAC works predominantly address either the transmitter-side design~\cite{24-TVT-ISAC} or the receiver-side processing~\cite{25-TCCN-ISAC-UL} in isolation, without considering their joint optimization. Since pilot and data transmissions jointly influence both communication reliability and sensing accuracy, coordinated resource allocation and beamforming design are essential for achieving favorable performance trade-offs. Nevertheless, the intrinsic coupling between communication and sensing objectives, together with the shared use of limited resources such as spectrum and transmit power, poses significant challenges for the efficient design of multi-user uplink ISAC systems.

To address these challenges, this paper investigates the problem of uplink achievable sum-rate maximization in a multi-user bistatic ISAC system while guaranteeing both communication and sensing QoS.
The considered system comprises multiple single-antenna uplink users and a multi-antenna BS that simultaneously performs user communication and target sensing.
Specifically, we propose a two-stage optimization framework: first, the pilot and data power allocations are optimized based on statistical channel state information (CSI); subsequently, leveraging minimum mean-squared error (MMSE) channel estimates, the receive beamforming is jointly designed to maximize the instantaneous achievable sum rate.
The main contributions of this paper are summarized as follows:

\begin{itemize}
    \item 
    We propose a two-stage design for a multi-user uplink ISAC system, consisting of long-term power allocation and instantaneous receive beamforming optimization. To this end, we derive a tractable lower bound on the average achievable sum rate and an exact expression for the target detection probability.
    
    \item 
    We formulate a joint pilot and data power allocation problem based on statistical CSI to maximize the derived sum-rate lower bound. An alternating optimization (AO) algorithm, incorporating successive convex approximation (SCA) and a penalty dual decomposition (PDD)-based gradient ascent method, is developed to obtain a stationary solution.
    
    \item 
    Based on the long-term power allocation, the BS optimizes receive beamforming to maximize the instantaneous sum rate. Specifically, the communication beamforming is optimized via SCA, whereas the target beamforming is derived in closed-form through eigenvalue decomposition (EVD).
    
    \item 
    Extensive simulation results are provided to validate the effectiveness of the proposed solution and to gain insights into how different design parameters influence the overall system performance. We also compare the proposed design with the maximal-ratio combining (MRC)-based and zero-forcing (ZF)-based receive beamforming schemes, demonstrating the superiority of the proposed approach.
\end{itemize}

The organization of this paper is as follows. The system model and the formulation of two optimization problems for the uplink multi-user ISAC system are presented in Section~II. Section~III develops the corresponding solution methods for these problems. Simulation results are provided in Section~IV to validate the effectiveness of the proposed designs. Finally, Section~V concludes the paper.

\textit{Notation}: 
Uppercase boldface and lowercase boldface letters represent matrices and vectors, respectively. 
$(\cdot)^{\mathsf T}$, $(\cdot)^{\mathsf H}$, and $ (\cdot)^{*} $ represent (ordinary) transpose, Hermitian transpose, and complex conjugate, respectively.
$ \left\| \cdot  \right\|_2 $ denotes the $ l_2 $-norm of a vector.
$ \mathbb{E}\left\{x\right\} $ denotes the expectation of $ x $. 
$ \mathbf{I} $ represents the identity matrix.
$ x \sim \mathcal{CN}(a, b) $ and $x \sim \mathcal{N}(a, b)$ mean that the scalar $x$ follows a complex Gaussian distribution and a real Gaussian distribution, respectively, with mean $a$ and variance $b$.
$ \text{tr}(\cdot) $ denotes the trace of a matrix.
$ \text{diag} (\cdot)  $ denotes a diagonal matrix.
$ \text{blkdiag} (\cdot) $ denotes a block diagonal matrix.
$\det \left( \mathbf{A } \right)$ denotes the determinant of $\mathbf{A}$.
$ \Re\{\cdot\} $ means taking the real part.
$\Pr(A|B)$ denotes the conditional probability of event A occurring given that event B has occurred. 
The complex-valued gradient of a real-valued function $f(\mathbf X)$, with respect to (w.r.t.) $\mathbf X$ is denoted by $\nabla_{\mathbf X} f(\mathbf X)$.
Finally, $ \mathbf{A} \in \mathbb{C}^{M\times N} $ indicates that $ \mathbf{A} $ is a complex-valued matrix with dimensions $ {M\times N} $.

\section{System Model and Problem Formulation}

\begin{figure}[t] 
	\centering
	\includegraphics[width=0.35\textwidth]{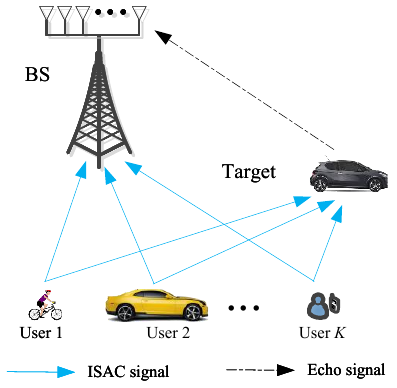}
	\caption{Multi-user uplink communication assisted target detection in ISAC systems.}
	\label{Fig_system}
\end{figure} 

\begin{figure}[t] 
	\centering
	\includegraphics[width=0.45\textwidth]{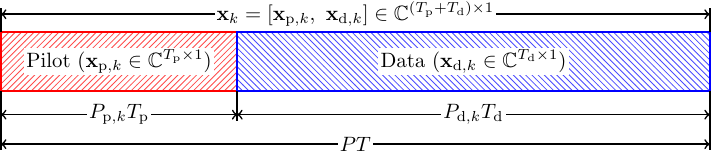}		
	\caption{Frame structure of the pilot and data transmitted by user-$k$.}
	\label{Fig_frame}
\end{figure}  

We consider a bistatic ISAC system, shown in Fig.~\ref{Fig_system}, consisting of a BS equipped with a uniform linear array (ULA) of ${{N}_{\text{b}}}$ receive antenna elements, $K$ single-antenna communication users, and a point-like target.
Specifically, the BS receives the users’ uplink communication signals for data decoding, while simultaneously leveraging their echo components for target detection. 
We assume that the sensing task aims to monitor a predefined region of interest~\cite{23-TWC-ISAC,24-TWC-ISAC-Pilot}.
Detection is performed at the BS via a binary hypothesis test (\emph{e.g.}, a Neyman--Pearson detector~\cite{98-SP-Dectect}) using the received signals corresponding to this region, without assuming knowledge of the instantaneous target location.

While this work focuses on a single-target model to establish the fundamental optimization framework and enable a clean analytical characterization of the sensing performance, the proposed approach can be extended to multi-target scenarios in a conceptually straightforward manner. Specifically, the received echo can be modeled as a linear superposition of multiple target components, while the overall optimization framework remains structurally similar. Nevertheless, such a generalization gives rise to several non-trivial challenges, including increased inter-target interference, coupling effects in the sensing metrics, spatial resolvability constraints, and more sophisticated detection rules, which are left for future investigation.

Moreover, we consider a discrete-time ISAC system operating over a quasi-static flat-fading channel, where the channel remains constant during each coherence block of $T$ symbol intervals and varies independently between consecutive blocks.
The transmitted ISAC signal of user-$k$ is given by $\mathbf{x}_k = \big[ \mathbf{x}_{\mathrm{p},k}^{\mathsf T},\, \mathbf{x}_{\mathrm{d},k}^{\mathsf T} \big]^{\mathsf T} \in \mathbb{C}^{T \times 1}, \quad k \in \mathcal{K}=\{1,\ldots,K\}$,
which comprises a pilot signal $\mathbf{x}_{\mathrm{p},k}$ and a data signal $\mathbf{x}_{\mathrm{d},k}$, as illustrated in Fig.~\ref{Fig_frame}. The pilot signal is assumed to be known at the BS, and the pilots transmitted by different users are mutually orthogonal. The durations of $\mathbf{x}_{\mathrm{p},k}$ and $\mathbf{x}_{\mathrm{d},k}$ are denoted by $T_{\mathrm{p}}$ and $T_{\mathrm{d}}$, respectively, with $T_{\mathrm{p}} + T_{\mathrm{d}} = T$. The corresponding transmit powers are denoted by $P_{\mathrm{p},k}$ and $P_{\mathrm{d},k}$. Accordingly, the total transmit energy of user-$k$ over one coherence block is given by $P_{\mathrm{p},k} T_{\mathrm{p}} + P_{\mathrm{d},k} T_{\mathrm{d}} \leq P T$, where $P$ denotes the maximum transmit power budget per coherence block.

\subsection{Channel Estimation}
The BS receives pilot signals transmitted by $K$ users to estimate the channel. Subsequently, the channel estimates are used to demodulate the users' message and also to detect the presence of the target. 
The pilot signal received by the BS after collecting $T_{\text{p}}$ pilot symbols can be written as
\begin{align}
	{{\mathbf{Y}}_{\text{p}}}=\sum\nolimits_{k\in \mathcal{K}}{\sqrt{{{P}_{\text{p},k}}}\left( {{\mathbf{h}}_{k}}+{{\mathbf{g}}_{k}} \right)}\mathbf{x}_{\text{p},k}^{\mathsf T}+{{\mathbf{N}}_{\text{p}}},
\end{align}
where $\mathbf{Y}_{\text{p}}\in\mathbb{C}^{N_{\text{b}}\times T_{\text{p}}}$ and
$\mathbf{x}_{\text{p},k}\in\mathbb{C}^{T_{\text{p}}\times1}$. 
Moreover, $\mathbf{h}_{k}\sim \mathcal{C}\mathcal{N}\left(\mathbf{0},\mathbf{R}_{\text{h},k}\right)\in\mathbb{C}^{N_{\text{b}}\times1}$ is the channel vector from user-$k$ to the BS,
and $\mathbf{g}_{k} = \sqrt{\alpha_{\text{RCS}}}\mathbf{a}\left(\theta_\text{t}\right)h_{\text{t},k} \sim \mathcal{C}\mathcal{N}\left(\mathbf{0},\mathbf{R}_{\mathrm{g},k}\right) \in\mathbb{C}^{N_{\text{b}}\times1}$ represents the echo channel vector for the path of user-$k$ $\rightarrow$  target $\rightarrow$  BS.\footnote{While the second-order statistic $\mathbf{R}_{\mathrm g,k}$ parameterizes the distribution under the competing hypotheses, it does not dictate the instantaneous target presence, which remains a random event inferred only from observations.} 
Here $\alpha_{\text{RCS}}$ denotes the reflection coefficient related to the radar-cross-section (RCS) of the target, and  
$\mathbf{a}\left(\theta_\text{t}\right)=[1,e^{-j2\pi d\sin\theta_\text{t}/\lambda},\ldots,e^{-j2\pi\left(N_{\text{b}}-1\right)d\sin\theta_\text{t}/\lambda}]^{\mathsf T}\in\mathbb{C}^{N_{\text{b}}\times1} $ represents the receive steering vector at the BS for the target, where $\theta_\text{t} $ denotes the azimuth angle of the target, $d$ denotes the spacing of array elements, and $\lambda$ denotes the wavelength. 
Additionally, $\mathbf{N}_{\text{p}}\in\mathbb{C}^{N_{\text{b}}\times T_{\text{p}}}$ is the additive white Gaussian noise (AWGN) at the BS during the pilot transmission phase, whose elements are independent and identically distributed (i.i.d.) and obey the distribution $\mathcal{C}\mathcal{N}(0,\sigma^2)$. 

By employing the known signal $\mathbf{x}_{\text{p},k}^{*}$ as the matched filter, the received signal at the BS after filtering can be expressed as
\begin{align}
	\label{eq:y_pk}
	{{\mathbf{y}}_{\text{p},k}}&={{\mathbf{Y}}_{\text{p}}}\mathbf{x}_{\text{p},k}^{*}=\sum\nolimits_{i\in \mathcal{K}}{\sqrt{{{P}_{\text{p},i}}}\left( {{\mathbf{h}}_{i}}+{{\mathbf{g}}_{i}} \right)\mathbf{x}_{\text{p},i}^{\mathsf T}\mathbf{x}_{\text{p},k}^{*}}+{{\mathbf{N}}_{\text{p}}}\mathbf{x}_{\text{p},k}^{*} 
	\notag  \\ 
	& =\sqrt{{{P}_{\text{p},k}}}{{T}_{\text{p}}}\left( {{\mathbf{h}}_{k}}+{{\mathbf{g}}_{k}} \right)+{{\mathbf{N}}_{\text{p}}}\mathbf{x}_{\text{p},k}^{*}, 
\end{align}
where ${{\mathbf{y}}_{\text{p},k}}\in {{\mathbb{C}}^{{{N}_{\text{b}}}\times 1}}$.
Based on the received signals in~\eqref{eq:y_pk}, the MMSE estimation is employed to estimate both $\mathbf{h}_{k}$ and $\mathbf{g}_{k}$. 
The MMSE estimator corresponding to $\mathbf{h}_{k}$ can be expressed as 
\begin{align}
	\mathbf{\hat{h}}_k=\sqrt{{{P}_{\text{p},k}}}{{\mathbf{R}}_{\text{h},k}}{{\left( {{P}_{\text{p},k}}{{T}_{\text{p}}}{\mathbf{R}}_k +\sigma^2\mathbf{I} \right)}^{-1}}{{\mathbf{y}}_{\text{p},k}},
\end{align}
where ${{\mathbf{R}}_k} = {{\mathbf{R}}_{\text{h},k}}+{{\mathbf{R}}_{\mathrm{g},k}}$. 
The covariance matrix of ${{\mathbf{\hat{h}}}_{k}}$ is obtained by
${{\mathbf{R}}_{\hat{\mathrm{h}},k}} ={{P}_{\text{p},k}}{{T}_{\text{p}}}{{\mathbf{R}}_{\text{h},k}}{{\left( {{P}_{\text{p},k}}{{T}_{\text{p}}} {\mathbf{R}}_k+\sigma^2\mathbf{I} \right)}^{-1}}{{\mathbf{R}}_{\text{h},k}} $. Thus, the covariance matrix of the channel estimation error ${\bm{\upepsilon }_{k} = {\mathbf{h}}_{k} - {\mathbf{\hat{h}}}_{k}}$ can be expressed as ${{\mathbf{R}}_{\upepsilon,k}}={{\mathbf{R}}_{\text{h},k}}-{{P}_{\text{p},k}}{{T}_{\text{p}}}{{\mathbf{R}}_{\text{h},k}}{{\left( {{P}_{\text{p},k}}{{T}_{\text{p}}}\mathbf{R}_k+\sigma^2\mathbf{I} \right)}^{-1}}{{\mathbf{R}}_{\text{h},k}}$. 
Similarly, the estimated channel of $\mathbf{g}_{k}$ is given by 
\begin{align}{{\mathbf{\hat{g}}}_{k}}=\sqrt{{{P}_{\text{p},k}}}{{\mathbf{R}}_{\mathrm{g},k}}{{\left( {{P}_{\text{p},k}}{{T}_{\text{p}}}{\mathbf{R}}_k +\sigma^2\mathbf{I} \right)}^{-1}}{{\mathbf{y}}_{\text{p},k}},
\end{align}
with the covariance matrix as ${{\mathbf{R}}_{\hat{\mathrm{g}},k}}={{P}_{\text{p},k}}{{T}_{\text{p}}}{{\mathbf{R}}_{\mathrm{g},k}}{{\left( {{P}_{\text{p},k}}{{T}_{\text{p}}} {\mathbf{R}}_k+\sigma^2\mathbf{I} \right)}^{-1}}{{\mathbf{R}}_{\mathrm{g},k}}$.
Then, the channel estimation error of $\mathbf{g}_{k}$ is ${{\bm{\upvarepsilon }}_{k} = {\mathbf{g}}_{k} - {\mathbf{\hat{g}}}_{k}}$, and its covariance matrix as ${{\mathbf{R}}_{\upvarepsilon,k}}={{\mathbf{R}}_{\mathrm{g},k}}-{{P}_{\text{p},k}}{{T}_{\text{p}}}{{\mathbf{R}}_{\mathrm{g},k}}{{\left( {{P}_{\text{p},k}}{{T}_{\text{p}}}\mathbf{R}_k+\sigma^2\mathbf{I} \right)}^{-1}}{{\mathbf{R}}_{\mathrm{g},k}}$.

\subsection{Communication}
Based on the estimated channels ${\mathbf{\hat{h}}}_{k}$ and ${\mathbf{\hat{g}}}_{k}$, the received signal at the BS after combining the signals during the data transmission phase can be written as
\begin{align}
	\label{eq:Yd}
	\mathbf{Y}_{\text{d}}=\sum\nolimits_{k\in \mathcal{K}}{\sqrt{{{P}_{\text{d},k}}}\left( {\mathbf{h}_{k}}+{\mathbf{g}_{k}} \right)\mathbf{x}_{\text{d},k}^{\mathsf T}}+{{\mathbf{N}}_{\text{d}}}, 
\end{align}
where ${{\mathbf{Y}}_{\text{d}}}\in {\mathbb{C}}^{N_\text{b} \times {T_\text{d}}}$, 
$\mathbf{x}_{\text{d},k}\in\mathbb{C}^{T_{\text{d}}\times1}$,  
and ${{\mathbf{N}}_{\text{d}}}\in {{\mathbb{C}}^{{{N}_{\text{b}}}\times {{T}_{\text{d}}}}}$ is the AWGN at the BS during the data transmission phase with i.i.d. elements obeying the distribution $\mathcal{C}\mathcal{N}(0,\sigma^2)$. Then, the received signal in~\eqref{eq:Yd} is processed
by the receive beamforming matrix $\mathbf{U}_{\text{c}} = \left[ \mathbf{u}_1,\cdots ,\mathbf{u}_K \right]\in {{\mathbb{C}}^{ {{N}_{\text{b}}} \times  K} }$,  where $\mathbf{u}_k \in {\mathbb{C}}^{N_\text{b}  \times  1}$ represents the receive beamforming vector for user-$k$.  
Specifically, the post-combining received signal at the BS, corresponding to user-$k$, can be given by~\cite{17-bjornson-massiveMIMO} 
\begin{align}
	\label{eq:yd,k}
	{{\mathbf{y}}_{\text{d},k}} =&\mathbf{u}_{k}^{\mathsf H}{{\mathbf{Y}}_{\text{d}}}=\underbrace{\sqrt{{{P}_{\text{d},k}}}\mathbb{E}\left\{ \mathbf{u}_{k}^{\mathsf H}\left( {{\mathbf{h}}_{k}}+{{\mathbf{g}}_{k}} \right) \right\}\mathbf{x}_{\text{d},k}^{\mathsf T}}_{\text{Desired signal over average channel}} 
	\notag \\ &
	+\underbrace{\sqrt{{{P}_{\text{d},k}}}\left[ \mathbf{u}_{k}^{\mathsf H}\left( {{\mathbf{h}}_{k}}+{{\mathbf{g}}_{k}} \right)-\mathbb{E}\left\{ \mathbf{u}_{k}^{\mathsf H}\left( {{\mathbf{h}}_{k}}+{{\mathbf{g}}_{k}} \right) \right\} \right]\mathbf{x}_{\text{d},k}^{\mathsf T}}_{\text{Desired signal over “unknown” channel}}
	\notag \\ &
	+\mathbf{u}_{k}^{\mathsf H}\sum\nolimits_{i\in \mathcal{K} \setminus \{k\}\;}{\sqrt{{{P}_{\text{d},i}}}\left( {{\mathbf{h}}_{i}}+{{\mathbf{g}}_{i}} \right)\mathbf{x}_{\text{d},i}^{\mathsf T}}+\mathbf{u}_{k}^{\mathsf H}{{\mathbf{N}}_{\text{d}}}.
\end{align}
For simplicity, we denote $\mathbf{z}_{k}={{\mathbf{g}}_{k}}+{{\mathbf{h}}_{k}},\ \forall k\in\mathcal{K}$.
Since deriving a closed-form expression for the ergodic rate is mathematically intractable, we adopt the well-known use-and-then-forget (UaTF) lower bound~\cite{17-bjornson-massiveMIMO} on the ergodic achievable rate, given by  
$\bar R_{k}^{\text{lb}}=\frac{{{T}_{\text{d}}}}{T}\log_2\left( 1+{{\bar\gamma }_{k}} \right)$
with 
\begin{align}
	\label{eq:bar gamma k}
	{{\bar{\gamma }}_{k}}\!=\!\frac{{{P}_{\text{d},k}}\left| \mathbb{E}\left\{ \mathbf{u}_{k}^{\mathsf H}{{\mathbf{z}}_{k}} \right\} \right|^{2}}{\sum\nolimits_{i\in \mathcal{K}}{{{P}_{\text{d},i}}\mathbb{E}\left\{ \left| \mathbf{u}_{k}^{\mathsf H}{{\mathbf{z}}_{i}} \right|^{2} \right\}}\!-\!{{P}_{\text{d},k}}\left| \mathbb{E}\left\{ \mathbf{u}_{k}^{\mathsf H}{{\mathbf{z}}_{k}} \right\} \right|^{2}\!+\! \sigma _k^{2}},
\end{align}
where $\sigma _k^{2} =  \left\| {{\mathbf{u}}_{k}} \right\|_{2}^{2}\sigma^2$.
Consequently, a lower bound on the ergodic sum rate of $K$ communication users can be given by 
\begin{align}
	& \bar R_{\text{sum}}^{\text{lb}}=\sum\nolimits_{k\in \mathcal{K}}{\bar R_{k}^{\text{lb}}}
	\notag \\ &
	=\sum\nolimits_{k\in \mathcal{K}} \frac{{{T}_{\text{d}}}}{T \ln (2)}\left[\ln \left(\sum\nolimits_{i\in \mathcal{K}}{{{P}_{\text{d},i}}\mathbb{E}\left\{ \left| \mathbf{u}_{k}^{\mathsf H}{{\mathbf{z}}_{i}} \right|^{2} \right\}}+ \sigma _k^{2} \right)\right. -
	\notag \\ & 
	\ \ 
	\left.  \ln\! \left(\sum\nolimits_{i\in \mathcal{K}}\!{{{P}_{\text{d},i}}\mathbb{E}\!\left\{ \left| \mathbf{u}_{k}^{\mathsf H}{{\mathbf{z}}_{i}} \right|^{2}\! \right\}}\!-\!{{P}_{\text{d},k}}\left| \mathbb{E}\left\{ \mathbf{u}_{k}^{\mathsf H}{{\mathbf{z}}_{k}} \!\right\} \right|^{2}\!+\! \sigma _k^{2} \!\right)\!\right].
\end{align}

\subsection{Sensing}
After obtaining the communication channel estimate $\mathbf{\hat{h}}_k$, interference cancellation (IC) is applied at the BS to remove the signal received via the direct user-BS links~\cite{21-JSTSP-ISAC}. Specifically, the BS reconstructs the estimated user signal using $\mathbf{\hat{h}}_k$ and then subtracts it from the received composite signal. Target detection is then performed on the resulting residual signal, which ideally contains the target echo, channel estimation error, and noise. Accordingly, the residual signal used for sensing can be expressed as:
\begin{align}
	y_{t}^{\text{sen}} = 
	\begin{cases}
		\mathbf{u}_\text{s}^{\mathsf H}\mathbf{  \left( \hat{G} + \Phi \right)}{{\mathbf{\Delta }}_{\text{p}}}{{\mathbf{x}}_{\text{p},t}}+\mathbf{u}_\text{s}^{\mathsf H}\mathbf{}{{\mathbf{n}}_{\text{p},t}}, t\in {{\mathcal{T}}_{\text{p}}}, 
		\\ 
		\mathbf{u}_\text{s}^{\mathsf H}\mathbf{ \left( \hat{G} + \Phi \right) }{{\mathbf{\Delta }}_{\text{d}}}{{\mathbf{x}}_{\text{d},t}}+\mathbf{u}_\text{s}^{\mathsf H}\mathbf{}{{\mathbf{n}}_{\text{d},t}}, t\in {{\mathcal{T}}_{\text{d}}},
	\end{cases}
\end{align}
where ${{\mathcal{T}}_{\text{p}}}=\left\{ 1,\cdots ,{{T}_{\text{p}}} \right\}$, 
${{\mathcal{T}}_{\text{d}}}=\left\{ T-{{T}_{\text{p}}}+1,\cdots ,T \right\}$, 
$\mathbf{x}_{\text{p},t}, \mathbf{x}_{\text{d},t}\in\mathbb{C}^{ K \times 1}$,
${{\mathbf{n}}_{\text{p},t}},{{\mathbf{n}}_{\text{d},t}}\in {{\mathbb{C}}^{{{N}_{\text{b}}}\times 1}}$, and $\mathbf{u}_\text{s}\in {\mathbb{C}}^{N_\text{b} \times 1}$ represents the receive beamforming vector for the target.
$\mathbf{\hat G }=\left[ {{\mathbf{\hat{g}}}_{1}},\cdots ,{{\mathbf{\hat{g}}}_{K}} \right]\in {{\mathbb{C}}^{{{N}_{\text{b}}}\times K}}$ is the stack of the estimation of $\mathbf g_k$ 
and 
$\mathbf{\Phi }=\left[ {\mathbf e_1},\cdots ,{\mathbf e_K} \right]\in {{\mathbb{C}}^{{{N}_{\text{b}}}\times K}}$ is the stack of the estimation error of $\mathbf z_k$ with $\mathbf e_k = \bm \upepsilon_k + \bm \upvarepsilon_k$.
${{\mathbf{\Delta }}_{\text{p}}}=\text{diag}(\sqrt{{{\mathbf{p}}_{\text{p}}}})$ and
${{\mathbf{\Delta }}_{\text{d}}}=\text{diag}(\sqrt{{{\mathbf{p}}_{\text{d}}}})$,
where
${{\mathbf{p}}_{\text{p}}}={{[{{P}_{\text{p},1}},\cdots ,{{P}_{\text{p},K}}]}^{\mathsf T}}\in {{\mathbb{R}}^{K\times 1}}$ and
${{\mathbf{p}}_{\text{d}}}={{[{{P}_{\text{d},1}},\cdots ,{{P}_{\text{d},K}}]}^{\mathsf T}}\in {{\mathbb{R}}^{K\times 1}}$. Additionally, the interference of ${\mathbf{\hat{h}}}_k$ has been removed through the IC technique.
For simplicity, we define $n_{\text{eff,p},t} \triangleq \mathbf{u}_\text{s}^{\mathsf H}\mathbf{\Phi}{{\mathbf{\Delta }}_{\text{p}}}{{\mathbf{x}}_{\text{p},t}}+\mathbf{u}_\text{s}^{\mathsf H}{{\mathbf{n}}_{\text{p},t}}$ and $n_{\text{eff,d},t} \triangleq  \mathbf{u}_\text{s}^{\mathsf H}\mathbf{\Phi}{{\mathbf{\Delta }}_{\text{d}}}{{\mathbf{x}}_{\text{d},t}}+\mathbf{u}_\text{s}^{\mathsf H}{{\mathbf{n}}_{\text{d},t}}$.

Then, the binary hypothesis testing formulation based on the presence (alternate hypothesis ${\mathcal{H}}_{1}$) or absence (null hypothesis ${\mathcal{H}}_{0}$) of the target can be expressed, respectively, as
\begin{align}
	\begin{aligned}
		& {\mathcal{H}}_{0} \text{: } y_{t}^{\text{sen}} = 
		\begin{cases}
			n_{\text{eff,p},t}, t\in {{\mathcal{T}}_{\text{p}}}, 
			\\ 
			n_{\text{eff,d},t}, t\in {{\mathcal{T}}_{\text{d}}},
		\end{cases}
		\\ 
		& {\mathcal{H}}_{1} \text{: } y_{t}^{\text{sen}} = 
		\begin{cases}
			\mathbf{u}_\text{s}^{\mathsf H}\mathbf{\hat{G}}{{\mathbf{\Delta }}_{\text{p}}}{{\mathbf{x}}_{\text{p},t}}+n_{\text{eff,p},t}, t\in {{\mathcal{T}}_{\text{p}}}, 
			\\ 
			\mathbf{u}_\text{s}^{\mathsf H}\mathbf{\hat{G}}{{\mathbf{\Delta }}_{\text{d}}}{{\mathbf{x}}_{\text{d},t}}+n_{\text{eff,d},t}, t\in {{\mathcal{T}}_{\text{d}}}.
		\end{cases}
	\end{aligned}
\end{align}

On the basis of the two hypotheses, the distributions followed by $y_{t}^{\text{sen}}$ can be written respectively as 
\begin{align}
	\begin{aligned}
		& {\mathcal{H}}_{0} \text{: } y_{t}^{\text{sen}}   \!  \sim  \! 
		\begin{cases}
			\mathcal{C}\mathcal{N} \!  \left( 0,\mathbf{u}_\text{s}^{\mathsf H}{{\mathbf{R}}_{\text{eff,p}}}{\mathbf{u}_\text{s}}\right), t\in {{\mathcal{T}}_{\text{p}}}, 
			\\ 
			\mathcal{C}\mathcal{N} \!  \left( 0,\mathbf{u}_\text{s}^{\mathsf H}{{\mathbf{R}}_{\text{eff,d}}}{\mathbf{u}_\text{s}} \right), t\in {{\mathcal{T}}_{\text{d}}},
		\end{cases}
		\\ 
		& {\mathcal{H}}_{1} \text{: } y_{t}^{\text{sen}}  \!  \sim  \! 
		\begin{cases}
			\mathcal{C}\mathcal{N} \! \left( \mathbf{u}_\text{s}^{\mathsf H}\mathbf{\hat{G}}{{\mathbf{\Delta }}_{\text{p}}}{{\mathbf{x}}_{\text{p},t}},\mathbf{u}_\text{s}^{\mathsf H}{{\mathbf{R}}_{\text{eff,p}}}{\mathbf{u}_\text{s}}\right)  , t\in {{\mathcal{T}}_{\text{p}}}, 
			\\ 
			\mathcal{C}\mathcal{N} \! \left( \mathbf{u}_\text{s}^{\mathsf H}\mathbf{\hat{G}}{{\mathbf{\Delta }}_{\text{d}}}{{\mathbf{x}}_{\text{d},t}},\mathbf{u}_\text{s}^{\mathsf H}{{\mathbf{R}}_{\text{eff,d}}}{\mathbf{u}_\text{s}} \right),  t\in {{\mathcal{T}}_{\text{d}}},
		\end{cases}
	\end{aligned}
\end{align}
where ${{\mathbf{R}}_{\text{eff,p}}}= \sum\nolimits_{k\in \mathcal{K}}{{{P}_{\text{p},k}}}{{\mathbf{R}}_{\text{err},k}}+\sigma^2\mathbf{I}$ and
${{\mathbf{R}}_{\text{eff,d}}}=\sum\nolimits_{k\in \mathcal{K}}{{{P}_{\text{d},k}}}{{\mathbf{R}}_{\text{err},k}}+\sigma^2\mathbf{I}$ with ${{\mathbf{R}}_{\text{err},k}}={{\mathbf{R}}_{\upepsilon,k}}+{{\mathbf{R}}_{\upvarepsilon,k}}$.
Thus, the probability density functions (PDFs) of the received signal under ${\mathcal{H}}_{0}$ and ${\mathcal{H}}_{1}$ can be represented respectively as~\cite{98-SP-Dectect}
\begin{align}
	& f\left( {{\mathbf{y}}^{\text{sen}}}\left| {{\mathcal{H}}_{0}} \right. \right)=\frac{{{e}^{-{{\left( {{\mathbf{y}}^{\text{sen}}} \right)}^{\mathsf H}}\mathbf{R}_{\text{eff}}^{-1}{{\mathbf{y}}^{\text{sen}}}}}}{{{\pi }^{T}}\det \left( {{\mathbf{R}}_{\text{eff}}} \right)}, 
	\notag \\ 
	& f\left( {{\mathbf{y}}^{\text{sen}}}\left| {{\mathcal{H}}_{1}} \right. \right)=\frac{{{e}^{-{{\left( {{\mathbf{y}}^{\text{sen}}}-\bm{\upmu } \right)}^{\mathsf H}}\mathbf{R}_{\text{eff}}^{-1}\left( {{\mathbf{y}}^{\text{sen}}}-\bm{\upmu } \right)}}}{{{\pi }^{T}}\det \left( {{\mathbf{R}}_{\text{eff}}} \right)},
	\label{eq:pdf}
\end{align}
where 
${{\mathbf{y}}^{\text{sen}}}=[y_{1}^{\text{sen}},\cdots ,y_{T}^{\text{sen}}]^{\mathsf T}$,
${{\mathbf{R}}_{\text{eff}}}=\text{blkdiag}(\mathbf{u}_{\text{s}}^{\mathsf H}{{\mathbf{R}}_{\text{eff},\text{p}}}{{\mathbf{u}}_{\text{s}}}{{\mathbf{I}}_{{{T}_{\text{p}}}}},\mathbf{u}_{\text{s}}^{\mathsf H}{{\mathbf{R}}_{\text{eff},\text{d}}}{{\mathbf{u}}_{\text{s}}}{{\mathbf{I}}_{{{T}_{\text{d}}}}})$, 
and $\bm{\upmu }=\left[ \bm{\upmu }_\text{p}, \bm{\upmu }_\text{d}\right]^{\mathsf T}$ with 
$\bm{\upmu }_\text{p}=\left[  \mathbf{u}_{\text{s}}^{\mathsf H}\mathbf{\hat{G}}{{\mathbf{\Delta }}_{\text{p}}}{{\mathbf{x}}_{\text{p},1}},\cdots ,\mathbf{u}_{\text{s}}^{\mathsf H}\mathbf{\hat{G}}{{\mathbf{\Delta }}_{\text{p}}}{{\mathbf{x}}_{\text{p},{{T}_{\text{p}}}}} \right]$ and $\bm{\upmu }_\text{d}=\left[\mathbf{u}_{\text{s}}^{\mathsf H}\mathbf{\hat{G}}{{\mathbf{\Delta }}_{\text{d}}}{{\mathbf{x}}_{\text{d},T-{{T}_{\text{p}}}+1}},\cdots ,\mathbf{u}_{\text{s}}^{\mathsf H}\mathbf{\hat{G}}{{\mathbf{\Delta }}_{\text{d}}}{{\mathbf{x}}_{\text{d},T}} \right]$.

Based on~\eqref{eq:pdf}, the log-likelihood ratio function can be expressed as
\begin{align}
	\ln\left(L\left( {{\mathbf{y}}^{\text{sen}}} \right)\right)&=\ln\left(\frac{f\left( {{\mathbf{y}}^{\text{sen}}}\left| {{\mathcal{H}}_{1}} \right. \right)}{f\left( {{\mathbf{y}}^{\text{sen}}}\left| {{\mathcal{H}}_{0}} \right. \right)}\right)
	\notag \\ &
	= 2\Re\left\{ {{\left( {{\mathbf{y}}^{\text{sen}}} \right)}^{\mathsf H}}\mathbf{R}_{\text{eff}}^{-1}\bm{\upmu } \right\}-{{\bm{\upmu }}^{\mathsf H}}\mathbf{R}_{\text{eff}}^{-1}\bm{\upmu }.
\end{align}
To simplify the calculation, $\psi\left({{\mathbf{y}}^{\text{sen}}}\right) = \Re\left\{ {{\left( {{\mathbf{y}}^{\text{sen}}} \right)}^{\mathsf H}}\mathbf{R}_{\text{eff}}^{-1}\bm{\upmu } \right\}$ is adopted as the decision variable. According to the Neyman-Pearson criterion, the likelihood ratio test (LRT) can be formulated as
$\psi\left({{\mathbf{y}}^{\text{sen}}}\right) \underset{{{\mathcal{H}}_{0}}}{\overset{{{\mathcal{H}}_{1}}}{\mathop{\gtrless }}} \kappa $,
where $\kappa$ denotes the detection  threshold.
Under hypothesis ${\mathcal{H}}_{0}$, $\psi\left({{\mathbf{y}}^{\text{sen}}}\right)\sim \mathcal{N}\left( 0,\frac{1}{2}{{\bm{\upmu }}^{\mathsf H}}\mathbf{R}_{\text{eff}}^{-1}\bm{\upmu } \right)$, and the probability of false alarm can be obtained as
\begin{align} 
	{{P}_{\text{FA}}}=\Pr\left( \psi \left( {{\mathbf{y}}^{\text{sen}}} \right)>\kappa \left| {{\mathcal{H}}_{0}} \right. \right)=Q\left( \frac{\sqrt{2}\kappa }{\sqrt{{{\bm{\upmu }}^{\mathsf H}}\mathbf{R}_{\text{eff}}^{-1}\bm{\upmu }}} \right),
\end{align}
where $ Q\left( x \right)=\frac{1}{\sqrt{2\pi }}\int_{x}^{\infty }{{{e}^{-{{t}^{2}}/2}}\, \mathrm{d}t}$ denotes the right tail probability function of the standard normal distribution. Particularly, we can obtain $\kappa =\frac{\sqrt{{{\bm{\upmu }}^{\mathsf H}}\mathbf{R}_{\text{eff}}^{-1}\bm{\upmu }}}{\sqrt{2}}{{Q}^{-1}}\left( {{P}_{\text{FA}}}  \right)$, where $Q^{-1}(x)$ is the inverse function of $Q(x)$.
Then, under the $\mathcal{H}_1$ hypothesis, $\psi\left({{\mathbf{y}}^{\text{sen}}}\right) \sim \mathcal{N}\left( {{\bm{\upmu }}^{\mathsf H}}\mathbf{R}_{\text{eff}}^{-1}\bm{\upmu },\frac{1}{2}{{\bm{\upmu }}^{\mathsf H}}\mathbf{R}_{\text{eff}}^{-1}\bm{\upmu } \right)$, and the target detection probability can be obtained as
\begin{align}
	{{P}_{\text{D}}}&=\Pr\left( \psi \left( {{\mathbf{y}}^{\text{sen}}} \right)>\kappa \left| {{\mathcal{H}}_{1}} \right. \right)=Q\left( \frac{\sqrt{2}\left( \kappa -{{\bm{\upmu}}^{\mathsf H}}\mathbf{R}_{\text{eff}}^{-1}\bm{\upmu} \right)}{\sqrt{{{\bm{\upmu}}^{\mathsf H}}\mathbf{R}_{\text{eff}}^{-1}\bm{\upmu}}} \right)
	\notag \\
	&=Q\left( {{Q}^{-1}}\left( {{P}_{\text{FA}}}  \right)-\sqrt{2\rho } \right),
	\label{eq:PD}
\end{align}
where $\rho ={{\bm{\upmu}}^{\mathsf H}}\mathbf{R}_{\text{eff}}^{-1}\bm{\upmu }$.
Therefore, we can conclude that ${\bm{\upmu}} \sim  \mathcal{C}\mathcal{N} \left(0,{{\mathbf{R}}_{\upmu}}\right)$, where ${{\mathbf{R}}_{\upmu}}=\text{blkdiag}(\mathbf{u}_{\text{s}}^{\mathsf H}{{\mathbf{R}}_{\upmu,\text{p}}}{{\mathbf{u}}_{\text{s}}}{{\mathbf{I}}_{{{T}_{\text{p}}}}},\mathbf{u}_{\text{s}}^{\mathsf H}{{\mathbf{R}}_{\upmu,\text{d}}}{{\mathbf{u}}_{\text{s}}}{{\mathbf{I}}_{{{T}_{\text{d}}}}})$ with 
${{\mathbf{R}}_{\upmu,\text{p}}}=\sum\nolimits_{k\in \mathcal{K}}{{{P}_{\text{p},k}}}{{\mathbf{R}}_{\hat{\mathrm{g}},k}}$ and
${{\mathbf{R}}_{\upmu,\text{d}}}=\sum\nolimits_{k\in \mathcal{K}}{{{P}_{\text{d},k}}}{{\mathbf{R}}_{\hat{\mathrm{g}},k}}$.
Then, the average target detection probability can be approximated as
\begin{align}
	\!\!\! {\bar{P}_{\text{D}}}={{\mathbb{E}}
	}\left\{ {{P}_{\text{D}}} \right\}\approx Q\left( {{Q}^{-1}}\left( {{P}_{\text{FA}}} \right)-\sqrt{2\text{tr}\left( {{\mathbf{R}}_{\upmu}}\mathbf{R}_{\text{eff}}^{-1} \right)} \right).
\end{align}

\subsection{Problem Formulation}
Note that since only statistical CSI is available before transmitting the pilot, we first formulate the problem of maximizing the UaTF lower bound on the ergodic achievable sum rate by deriving the optimal pilot and data power allocation, and receive beamforming at the BS, while guaranteeing the communication and sensing QoS. To obtain a stationary solution to this problem, we adopt an MRC-based receive beamforming approach, subsequently solving the optimal pilot and data power allocation via the AO algorithm. After obtaining the optimal power allocation scheme, we next obtain the channel estimates using MMSE estimation, and then formulate another problem of instantaneous achievable sum rate maximization via optimal beamforming design at the BS. 
With the given background, the problem of maximizing the lower bound on the ergodic sum achievable rate is formulated as
\begin{subequations}
	\label{P1}
	\begin{align}
		\left(\mathcal{P}_1 \right)  \max_{\mathbf{U}, \mathbf{P}}\quad &{\bar R}_{\text{sum}}^{\text{lb}}\left( \mathbf{U},\mathbf{P} \right) \label{P1a}\\ 
		\text{s.t.} \quad & 
		{{T}_{\text{p}}}{{P}_{\text{p},k}}+{{T}_{\text{d}}}{{P}_{\text{d},k}}\le PT,\forall k\in \mathcal{K}, \label{P1b} \\ 
		& {{P}_{\text{p},k}},{{P}_{\text{d},k}}\ge 0,\forall k\in \mathcal{K}, \label{P1c}\\ 
		& {\bar R}_{k}^{\text{lb}}\left( \mathbf{U},\mathbf{P} \right)\ge {{R}_{\text{th},k}},\forall k\in \mathcal{K}, \label{P1d}\\ 
		& {\bar{P}_{\text{D}}}\left( \mathbf{U},\mathbf{P} \right)\ge {{P}_{\text{D, th}}}, \label{P1e}
		\\ 
		& \left\|\mathbf{u}_{\text{s}} \right\|_2^2=1,\left\|\mathbf{u}_{k} \right\|_2^2=1,\forall k\in \mathcal{K}, \label{P1f}
	\end{align}
\end{subequations}
\leavevmode\unskip
where $\mathbf{U} = \left[ \mathbf{U}_\text{c}, \mathbf{u}_\text{s} \right] \in {{\mathbb{C}}^{ {{N}_{\text{b}}} \times \left( K+1 \right) }}$, $\mathbf{P}=[{{\mathbf{p}}_{\text{p}}},{{\mathbf{p}}_{\text{d}}}]\in {{\mathbb{R}}^{K\times 2}}$.
In ($\mathcal P_1$), the objective in~\eqref{P1a} represents the lower bound on the average achievable sum rate,~\eqref{P1b} represents the constraint on total energy transmitted during the channel coherence time $T$,~\eqref{P1c} enforces the non-negativity of the transmit power of user-$k$ during pilot and data transmission phases,~\eqref{P1d} ensures that the lower-bound on the average achievable rate for user-$k$ is greater than or equal to the threshold $R_{\mathrm{th},k}$, the quality-of-service on sensing is constrained by~\eqref{P1e} with $P_{\mathrm{D,th}}$ being the threshold on target detection probability, and~\eqref{P1f} ensures that the communication and sensing receive beamforming vectors are unit-norm.

\section{Proposed Solution}
It is not difficult to see that $\left(\mathcal{P}_1 \right) $ is non-convex due to the variable coupling in~\eqref{P1a},~\eqref{P1d}, and~\eqref{P1e}, and the constraints~\eqref{P1d},~\eqref{P1e}, and~\eqref{P1f} are particularly difficult to handle. 
It is noteworthy that due to the presence of $Q$-function (and its inverse), the sensing constraint in~\eqref{P1e} is intractable. 
To tackle this issue, we first note that $P_{\mathrm D}$ is a non-decreasing function of $\rho$. Thus, the constraint on $P_\text{D}$ can be mapped to a constraint on $\rho$. Furthermore, $\rho$ represents the energy of the signal ${{\mathbf{y}}^{\text{sen}}}$ after noise whitening, which can be regarded as the effective signal-to-noise ratio (SNR). 
Therefore, to simplify the calculation, we introduce the ergodic sensing signal-to-interference-plus-noise ratio (SINR), which can be written as
\begin{align}
	\label{eq:gammas}
	{\bar{\gamma }_{\text{s}}}\!=\!\frac{\sum\nolimits_{k\in \mathcal{K}}\left( {{T}_{\text{p}}}{{P}_{\text{p},k}}\!+\!{{T}_{\text{d}}}{{P}_{\text{d},k}} \right){{\mathbf{u}}_{\text{s}}^{\mathsf H}}{{\mathbf{R}}_{\hat{\mathrm{g}},k}}\mathbf{u}_{\text{s}}}{\sum\nolimits_{k\in \mathcal{K}}{\left( {{T}_{\text{p}}}{{P}_{\text{p},k}}\!+\!{{T}_{\text{d}}}{{P}_{\text{d},k}} \right){{\mathbf{u}}_{\text{s}}^{\mathsf H}}{{\mathbf{R}}_{\text{err},k}}\mathbf{u}_{\text{s}}}\!+\!T\left\| {{\mathbf{u}}_{\text{s}}} \right\|_{2}^{2}\sigma^2}.
\end{align}
By exploiting the linear dependency of $\rho$ and the sensing SINR, and based on~\eqref{eq:PD}, the value of $\gamma_{\mathrm{s,th}}$ corresponding to ${{P}_{\text{D, th}}}$ can be found, where $\gamma_{\mathrm{s,th}}$ denotes the threshold for the sensing SINR. 
Then, the constraint on $P_{\mathrm D}$ in~\eqref{P1e} is replaced by the constraint on ${\bar{\gamma }_{\text{s}}}$.
As a consequence, $\mathcal{P}_1$ can be rewritten as
\begin{align}
	\label{P2}
	\boxed{\left(\mathcal{P}_2 \right) \max_{\mathbf{U}, \mathbf{P}}\,\ \big\{{\bar R}_{\text{sum}}^{\text{lb}}\big|\eqref{P1b}-\eqref{P1d},~\eqref{P1f},\bar{\gamma} _\text{s}\ge { \gamma_{\mathrm{s,th}}} \big\}.}
\end{align}

Then, to tackle this challenging non-convex optimization problem $\left(\mathcal{P}_2 \right) $, we adopt the well-known AO approach, where we first optimize the variable $\mathbf{U}$ while keeping $\mathbf P$ fixed, and then optimize $\mathbf{P}$ for a fixed $\mathbf{U}$. 

\subsection{Receive Beamforming Design Before Transmitting Pilots}

Recall that we adopt a two-stage design framework. In the first stage, based on statistical channel information, we employ an MRC-based receive combining scheme for the communication users and an EVD-based scheme for the sensing target to derive a tractable expression for the average achievable performance. Under these fixed combining structures, the optimal pilot and data power allocation scheme is determined by solving problem $(\mathcal{P}_2)$. 

In the second stage, once the optimal power allocation is obtained, the users transmit pilot signals to enable channel estimation at the BS. Based on the estimated instantaneous CSI, the problem of instantaneous achievable sum-rate maximization is then solved to obtain the optimal receive combining vectors.

With the power allocation matrix $\mathbf{P}$ fixed, the MRC-based receive combiner for the $k$-th communication user is given by
$\mathbf{u}_{k}=(\hat{\mathbf{h}}_k+{{{\mathbf{\hat{g}}}}_{k}}) / \sqrt{\mathbb{E} \big\{\|\hat{\mathbf{h}}_k+{{{\mathbf{\hat{g}}}}_{k}}\|_2^2\big\}}= \mathbf{\hat z}_{k} /  \sqrt{\mathbb{E} \left\{\|\mathbf{\hat z}_{k}\|_2^2\right\}},\ \forall k\in\mathcal{K}$.
Using MMSE channel estimation, a closed-form expression for $\bar{\gamma}_{k}$ is derived in Theorem~\ref{thm:bar gamma k}.
\begin{Theorem}
	\label{thm:bar gamma k}
	A closed-form expression for ${{\bar{\gamma }}_{k}}$ is given by
	\begin{align}
		\label{eq:bar gamma k 2}
		{{\bar{\gamma }}_{k}}=\frac{{{P}_{\text{d},k}}\text{tr}\left( {{\mathbf{R}}_{\text{est},k}} \right)}{\sum\nolimits_{{i\in \mathcal{K}}\;}{{{P}_{\text{d},i}}\frac{\text{tr}\left( {{\mathbf{R}}_{i}}{{\mathbf{R}}_{\text{est},k}} \right)}{\text{tr}\left( {{\mathbf{R}}_{\text{est},k}} \right)}}+\sigma^2},
	\end{align}
	where
	${{\mathbf{R}}_{k}}={{\mathbf{R}}_{\text{h},k}}+{{\mathbf{R}}_{\mathrm{g},k}}$
	and
	${{\mathbf{R}}_{\text{est},k}}={{\mathbf{R}}_{\hat{\mathrm{h}},k}}+{{\mathbf{R}}_{\hat{\mathrm{g}},k}}={{P}_{\text{p},k}}{{T}_{\text{p}}}{{\mathbf{R}}_{k}}{\left( {{P}_{\text{p},k}}{{T}_{\text{p}}}{{\mathbf{R}}_{k}}+\sigma^2\mathbf{I} \right)^{-1}}{{\mathbf{R}}_{k}}$.
\end{Theorem}
\begin{IEEEproof}
	Please see Appendix~\ref{sec:Proof-bar-gamma-k}.
\end{IEEEproof}

At the same time, for sensing, (\ref{eq:gammas}) can be rewritten as
\begin{align}
	\label{eq:gammas 2}
	{{\bar\gamma }_{\text{s}}}
    =\frac{{{\mathbf{u}}_{\text{s}}^{\mathsf H}}{{\mathbf{\Omega }}_{1}}\mathbf{u}_{\text{s}}}{{{\mathbf{u}}_{\text{s}}^{\mathsf H}}\left( {{\mathbf{\Omega }}_{2}}+T\sigma^2  \mathbf{I}\right)\mathbf{u}_{\text{s}}},
\end{align}
where ${{\mathbf{\Omega }}_{1}}=\sum\nolimits_{k\in \mathcal{K}}{\left( {{T}_{\text{p}}}{{P}_{\text{p},k}}+{{T}_{\text{d}}}{{P}_{\text{d},k}} \right){{\mathbf{R}}_{\hat{\mathrm{g}},k}}}$ and
${{\mathbf{\Omega }}_{2}}=\sum\nolimits_{k\in \mathcal{K}}{\left( {{T}_{\text{p}}}{{P}_{\text{p},k}}+{{T}_{\text{d}}}{{P}_{\text{d},k}} \right){{\mathbf{R}}_{\text{err},k}}}$.
Then, according to the concept of the generalized Rayleigh quotient, the optimal receive beamforming vector for the target can be given by~\cite{25-ICC-ISCPT}
\begin{align}
	\label{eq:us-opt}
	\mathbf{u}_{\text{s}}^{\text{opt}}=\boldsymbol{\lambda}_{\max }\big\{ {{\left( {{\mathbf{\Omega }}_{2}}+T\sigma^2 \mathbf{I}\right)}^{-1}}{{\mathbf{\Omega }}_{1}} \big\},
\end{align}
where $\boldsymbol{\lambda}_{\max }\{ \mathbf X \}$ represents the eigenvector associated with the largest eigenvalue of the square matrix $\mathbf X$.
It is easy to note that following the MRC-based scheme, $\mathbf u_k \ \forall k \in \mathcal K$ is unit-norm. Also, following the \emph{orthonormality} property of the eigenvectors, $\mathbf u_{\mathrm s}^{\mathrm{opt}}$ is also unit-norm. This means using MRC and EVD schemes,~\eqref{P1f} is satisfied by default.

{   \newcounter{TempEqCnt3} 
	\setcounter{TempEqCnt3}{\value{equation}} 
	\setcounter{equation}{28} 
	\setlength{\abovedisplayskip}{0.1cm}
	\begin{figure*}[tb] 	
		\vspace{-0.6cm}
		\hrulefill	
		\begin{align}
			\label{eq:grad_Pd_closed}
			& \nabla_{P_{\mathrm{p},\ell}}\mathcal{L}_{\bm{\upomega},\xi}(\mathbf{p}_{\text{p}},\bm{\uptau})=\frac{{{T}_{\text{d}}}}{T\ln(2)}\sum\nolimits_{k\in \mathcal{K}}
			\left( 1+ \frac{\omega_{k}}{R_{\mathrm{th},k}}+\frac{{f_k}\left(\mathbf{p}_{\text{p}},{{\tau }_{k}}\right)}{R_{\mathrm{th},k} \xi} \right)
			\left[-\frac{1}{t_{\text{D},k}}\left({{\sum\nolimits_{i\in\mathcal{K}}{{P}_{\text{d},i}}\text{tr}\left( {\mathbf{R}}_i \acute{\mathbf{R}_{\mathrm{est},k}} \right)} + \sigma^2\text{tr}\left( \acute{\mathbf{R}_{\mathrm{est},k}}\right)}\right)\right.
			\notag \\ &
			\! +\! \left. \frac{1}{t_{\text{N},k}}\!\left(\!{2P_{\mathrm{d},k}\text{tr}\left( {{\mathbf{R}}_{\text{est},k}} \right)\text{tr}\left( \acute{\mathbf{R}_{\mathrm{est},k}}\right) \!+ \!\sum\nolimits_{i\in\mathcal{K}}{{P}_{\text{d},i}}\text{tr}\Big( {\mathbf{R}}_i \acute{\mathbf{R}_{\mathrm{est},k}} \Big) \!+ \!\sigma^2\text{tr}\Big( \acute{\mathbf{R}_{\mathrm{est},k}}\Big)}\!\right)\!\right] \! - \! \left(\omega_{\mathrm{s}}+\frac{{{f}_{\text{s}}}\left(\mathbf{p}_{\text{p}},{{\tau }_\text{s}}\right)}{\xi}\right)\frac{t_{\text{N,s}}\acute{t_{\text{D,s}}}\!-\!t_{\text{D,s}}\acute{t_{\text{N,s}}}}{\gamma _{\text{s,th}} t_{\text{D,s}}^2},
		\end{align}
		where $\acute{\mathbf{R}_{\mathrm{est},k}}$, $\acute{t_{\text{N,s}}}$, and $\acute{t_{\text{D,s}}}$ are defined in Appendix~\ref{sec:Proof-grad_Pd}.
		
		\hrulefill   
	\end{figure*}
	
	\setcounter{equation}{\value{TempEqCnt3}}}
    
\subsection{Optimizing Power Allocation}
After obtaining the receive combining filter before pilot transmission, we move on to the problem of optimal power allocation. For a fixed $\mathbf U$, $(\mathcal{P}_3)$ boils down to the following optimization problem:
\begin{subequations}
	\label{P3}
	\begin{align}
		\left(\mathcal{P}_{3}\right)  \max_{\mathbf{P}\ge 0}\quad &{\bar R}_{\text{sum}}^{\text{lb}}\left( \mathbf{P} \right) \label{P3a}\\ 
		\text{s.t.}\quad &{{T}_{\text{p}}}{{P}_{\text{p},k}}+{{T}_{\text{d}}}{{P}_{\text{d},k}}\le PT,\forall k\in \mathcal{K}, \label{P3b}\\
		&  {\bar R}_{k}^{\text{lb}}\left( \mathbf{P} \right)\ge {{R}_{\text{th},k}},\forall k\in \mathcal{K}, \label{P3c}\\ 
		& \bar{\gamma} _\text{s}\left( \mathbf{P} \right)\ge \gamma _{\text{s,th}}. \label{P3d} 
	\end{align}
\end{subequations}
Obviously, the objective function of $\left(\mathcal{P}_3 \right) $ is non-convex and the matrix inverse is coupled with the variable $\mathbf{p}_{\mathrm{p}}$, which further introduces non-linear complexity. Hence, we employ  an AO-based approach of fixing variable $\mathbf{p}_{\mathrm{d}}$ to optimize variable $\mathbf{p}_{\mathrm{p}}$ and then optimizing variable $\mathbf{p}_{\mathrm{d}}$ based on the optimized variable $\mathbf{p}_{\mathrm{p}}$.
\subsubsection{Optimal pilot power allocation  $\mathbf{p}_{\mathrm{p}}$}
The problem of obtaining optimal $\mathbf{p}_{\mathrm{p}}$ for a fixed $\mathbf{p}_{\mathrm{d}}$ is more challenging to solve since the problem is highly non-convex, and convexifying the problem is a tedious task. 
Therefore, we rely on the PDD-based gradient ascent approach to obtain a stationary solution to the problem. For this purpose, we introduce the slack variables $\bm{\uptau }={{[{{\tau }_{1}},\cdots,{{\tau }_{K}},{{\tau }_{s}}]}^{\mathsf T}}\in \mathbb{R}_{\geq 0}^{(K+1)\times 1}$ to convert the inequality constraints~\eqref{P3c} and~\eqref{P3d} into the equality constraints, which can be expressed respectively as
\begin{align}
	&{{f}_{k}}\left(\mathbf{p}_{\text{p}},{{\tau }_{k}}\right)= 1-\frac{1}{{R}_{\text{th},k}}{\bar R}_{k}^{\text{lb}}\left( \mathbf{p}_{\text{p}} \right)+{{\tau }_{k}}=0, \\
	&{{f}_{\text{s}}}\left(\mathbf{p}_{\text{p}},{{\tau }_{\text{s}}} \right)= 1 -\frac{1}{\gamma _{\text{s,th}}}{\bar{\gamma }_{\text{s}}}\left( \mathbf{p}_{\text{p}} \right)+{{\tau }_{\text{s}}}=0. \label{eq:fs} 
\end{align}
Hence, the augmented Lagrangian function for $\left(\mathcal{P}_3 \right) $ with a fixed $\mathbf{p}_{\mathrm{d}}$ can be formulated as 
\begin{align}
	\label{eq:L}
	{{\mathcal{L}}_{\bm{\upomega },\xi }}\!\left(\mathbf{p}_{\text{p}},\bm{\uptau } \right)=& {{\bar R}_{\text{sum}}^{\text{lb}}}\left(\mathbf{p}_{\text{p}}\right)\!-\! \sum\nolimits_{k\in \mathcal{K}}\!{{{\omega }_{k}}{f_k}\left(\mathbf{p}_{\text{p}},{{\tau }_{k}}\right)}\!-\!{{\omega }_{\text{s}}}{{f}_{\text{s}}}\left(\mathbf{p}_{\text{p}},{{\tau }_{\text{s}}} \right)
	\notag \\
	&\!-\!\frac{1 }{2 \xi} \left(\sum\nolimits_{k\in \mathcal{K}}{f_k^2}\left(\mathbf{p}_{\text{p}},{{\tau }_{k}}\right)\!+\!{{f}_{\text{s}}^2}\left(\mathbf{p}_{\text{p}},{{\tau }_\text{s}}\right)\right),
\end{align}
where 
$\bm{\upomega } = {{[{{\omega }_{1}},\cdots ,{{\omega }_{K}},{{\omega }_{\text{s}}}]}^{\mathsf T}}\in {{\mathbb{R}}^{(K+1)\times 1}}$ represents the vector containing Lagrangian multipliers and $\xi > 0$ represents the penalty parameter.
Finally, $\left(\mathcal{P}_3 \right) $ with a fixed $\mathbf{p}_{\mathrm{d}}$ can be transformed into the following equivalent optimization problem
\begin{align}
	\label{P4}
	\boxed{\left(\mathcal{P}_4 \right) \max_{\mathbf{p}_{\text{p}},\bm{\uptau }} \big\{ {{\mathcal{L}}_{\bm{\upomega },\xi }}\left(\mathbf{p}_{\text{p}},\bm{\uptau } \right)\big|\bm{\uptau } \geq 0,~\eqref{P3b} \big\}.}
\end{align}
Then, $\left(\mathcal{P}_4 \right) $ can be solved by iteratively updating $\mathbf{p}_{\text{p}}$, $\bm{\uptau }$, $\bm{\upomega }$, and $\xi$ until the convergence criterion is satisfied. The update of $\mathbf{p}_{\text{p}}$ at the $(\imath)$-th iteration can be expressed as
\begin{align}
	\label{eq:Update_Pd}
	\mathbf{p}_{\text{p}}^{(\imath)}=\mathbf{p}_{\text{p}}^{(\imath-1)}+\delta\nabla_{\mathbf{p}_{\text{p}}}\mathcal{L}_{\bm{\upomega},\xi}\big(\mathbf{p}_{\text{p}}^{(\imath-1)},\bm{\uptau}^{(\imath-1)}\big),
\end{align}
where $\delta$ denotes the step size. 
The optimal step-size $\delta$ is obtained using the backtracking line search following Armijo's condition~\cite{66-Mini}.
In addition, the closed-form expression for $\nabla_{\mathbf{p}_{\text{p}}}\mathcal{L}_{\bm{\upomega},\xi}(\mathbf{p}_{\text{p}},\bm{\uptau})$ is derived in Theorem~\ref{thm:grad_Pd}. 
\begin{Theorem}
	\label{thm:grad_Pd}
	The gradient $\nabla_{\mathbf{p}_{\text{p}}}\mathcal{L}_{\bm{\upomega},\xi}(\mathbf{p}_{\text{p}},\bm{\uptau})$ w.r.t. $\mathbf{p}_{\mathrm{p}}$ can be given by $\nabla_{\mathbf{p}_{\text{p}}}\mathcal{L}_{\bm{\upomega},\xi}(\mathbf{p}_{\text{p}},\bm{\uptau})=\big[\nabla_{\mathbf{p}_{\text{p},1}}\mathcal{L}_{\bm{\upomega},\xi}(\mathbf{p}_{\text{p}},\bm{\uptau}),\ldots,\nabla_{\mathbf{p}_{\text{p},K}}\mathcal{L}_{\bm{\upomega},\xi}(\mathbf{p}_{\text{p}},\bm{\uptau})\big]^{\mathsf T}$. 
	Then, a closed-form expression for $\nabla_{P_{\mathrm{p},\ell}}\mathcal{L}_{\bm{\upomega},\xi}(\mathbf{p}_{\text{p}},\bm{\uptau})$ is given by~\eqref{eq:grad_Pd_closed}, which is shown at the top of this page.
\end{Theorem}
\begin{IEEEproof}
	Please see Appendix~\ref{sec:Proof-grad_Pd}.
\end{IEEEproof}
The update of other parameters at the $(\imath)$-th iteration can be expressed as
\setcounter{equation}{29}
\begin{align}
	&\tau_{k}^{(\imath)}=
	\max\left\{0,\frac{1}{R_{\mathrm{th},k}}{\bar  R_{k}^{\text{lb}}\big(\mathbf{p}_{\text{p}}^{(\imath)}\big)}-1-\omega_{k}\xi\right\},
	\label{eq:tau_c}
	\\
	&\tau_{\mathrm{s}}^{(\imath)}=
	\max\left\{0,\frac{1}{\gamma_{\mathrm{s,th}}}{\bar\gamma_{\mathrm{s}}\big(\mathbf{p}_{\text{p}}^{(\imath)}\big)}-1-\omega_{\mathrm{s}} \xi\right\},
	\label{eq:tau_s}
	\\
	&\omega_{k} \leftarrow \omega_{k}+\frac{1}{\xi}  f_{k}\big(\mathbf{p}_{\text{p}}^{(\imath)},\tau_k^{(\imath)}\big), \label{eq:omega_c}
	\\
	&
    \omega_{\mathrm{s}} \leftarrow  \omega_{\mathrm{s}} + \frac{1}{\xi} f_{\mathrm{s}}\big(\mathbf{p}_{\text{p}}^{(\imath)},\tau_{\text{s}}^{(\imath)}\big), \label{eq:omega_s}
	\\
	&\xi \leftarrow \eta\xi,  \label{eq:xi}  
\end{align}
where $\eta$ denotes the growth factor of the penalty parameter.
    
\subsubsection{Optimal data-transmission power allocation $\mathbf{p}_{\mathrm{d}}$}
For the case when $\mathbf{p}_{\mathrm{p}}$ is kept fixed, we can obtain a \emph{concave} lower-bound on~\eqref{P3a} as
\begin{align}
	\label{eq:P2-obj-DC}
	&{\bar R}_{\text{sum}}^{\text{lb}}\left( \mathbf{p}_\text{d} \right) 
	= \frac{{{T}_{\text{d}}}}{T\ln\left(2\right)}\sum\nolimits_{k\in \mathcal{K}}\ln\left( 1+ \frac{{{\varphi }_{k} {P}_{\text{d},k}}}{\sum\nolimits_{{i\in \mathcal{K}}}{{\varsigma }_{k,i}{{P}_{\text{d},i}}}+\sigma^2}\right)
	\notag \\ 
	&\overset{\texttt{(a)}}{\mathop{=}}  \frac{T_{\mathrm{d}}}{T\ln(2)} \sum\nolimits_{k\in\mathcal{K}}\left[\ln\left(\sigma^2+{\varphi }_{k} P_{\mathrm{d},k}   +\sum\nolimits_{i\in\mathcal{K}}{\varsigma }_{k,i}{{P}_{\text{d},i}}\right)\right.
	\notag \\ 
	& \ \ \ 
	\left.-\ln\left(\sigma^2+\sum\nolimits_{i\in\mathcal{K}}{\varsigma }_{k,i}{{P}_{\text{d},i}}\right)\right]
	\notag \\ 
	& \ge \frac{T_{\mathrm{d}}}{T\ln(2)} \sum\nolimits_{k\in\mathcal{K}}\left[\ln\left(\sigma^2+{\varphi }_{k} P_{\text{d},k}   +\sum\nolimits_{i\in\mathcal{K}}{\varsigma }_{k,i}{{P}_{\text{d},i}}\right)\right.
	\notag \\ 
	& \ \ \ 
	\left.-\ln\left(\sigma^2+\sum\nolimits_{i\in\mathcal{K}}\!{\varsigma }_{k,i}{{P}_{\text{d},i}^{(l-1)}}\right)
	\!- \!\frac{\sum\nolimits_{i\in\mathcal{K}}\!{\varsigma }_{k,i}\!\left(P_{\text{d},i} \! - \! P_{\text{d},i}^{(l-1)}\right)\!}{\sigma^2+\sum\nolimits_{i\in\mathcal{K}}\!{\varsigma }_{k,i}P_{\text{d},i}^{(l-1)}} \right].
	\notag \\ 
	& \triangleq 
	\sum\nolimits_{k\in\mathcal{K}} {\bar R}_{k}^{\text{lb}}\left( \mathbf{p}_\text{d};\mathbf{p}_\text{d}^{(l-1)} \right) ={\bar R}_{\text{sum}}^{\text{lb}}\left( \mathbf{p}_\text{d};\mathbf{p}_\text{d}^{(l-1)} \right),
\end{align}
where ${{\varphi }_{k}} = \text{tr}(\mathbf{R}_{\mathrm{est},k})$,
${{\varsigma }_{k,i}}=\text{tr}(\mathbf{R}_{i}\mathbf{R}_{\mathrm{est},k})/\text{tr}(\mathbf{R}_{\mathrm{est},k}), \forall k,i\in\mathcal{K}$
and $\mathbf{p}_{\mathrm{d}}^{(l-1)}$ is a feasible solution in the $(l-1)$-th iteration.
Note that ${\bar R}_{\text{sum}}^{\text{lb}}\left( \mathbf{p}_\text{d};\mathbf{p}_\text{d}^{(l-1)} \right)$ is a \textit{concave lower-bound} on $\bar R_{\mathrm{sum}}^{\text{lb}}(\mathbf{p}_{\mathrm{d}})$ obtained by linearizing the second term in the right-hand side of $\texttt{(a)}$. 
Since~\eqref{P3b} is already convex, we focus on the constraint in~\eqref{P3c}, which can be further written as 
\begin{align}
	\label{eq:Rkpd-DCP}
	& \ \  \ 1+ \frac{{{\varphi }_{k} {P}_{\text{d},k}}}{\sum\nolimits_{{i\in \mathcal{K}}}{{\varsigma }_{k,i}{{P}_{\text{d},i}}}+\sigma^2}\geq e^{\frac{R_{\mathrm{th},k}T\ln(2)}{T_{\mathrm{d}}}} \triangleq \gamma_{\mathrm{th},k}
	\notag \\  &
	\! \!\!\! \Rightarrow \!
	\sigma^2\!+\!{\varphi }_{k} P_{\text{d},k} \!+\!\sum\limits_{i\in\mathcal{K}}{\varsigma }_{k,i}{{P}_{\text{d},i}} \! \geq \!
	\gamma_{\mathrm{th},k} \!\Bigg(\!\sigma^2\! +\!\sum\limits_{i\in\mathcal{K}}{\varsigma }_{k,i}{{P}_{\text{d},i}}\!\Bigg)\!,
\end{align}
which is a linear constraint w.r.t. $\mathbf{p}_{\mathrm{d}}$. 
Next, we are now left with the non-convex constraint in~\eqref{P3d}, which can be further written as 
\begin{align}
	\label{eq:sensingSINR-DCP}
	&T_{\mathrm{d}}\sum\nolimits_{k\in\mathcal{K}}P_{\mathrm{d},k}\mathbf{u}_{\mathrm{s}}^{\mathsf H}\big(\mathbf{R}_{\mathbf{\hat g},k}-\gamma_{\mathrm{s,th}}\mathbf{R}_{\mathrm{err},k}\big)\mathbf{u}_{\mathrm{s}}-\gamma_{\mathrm{s,th}}T\sigma^2\|\mathbf{u}_{\mathrm{s}}\|_2^{2}
	\notag \\
	&+T_{\mathrm{p}}\sum\nolimits_{k\in\mathcal{K}}P_{\mathrm{p},k}\mathbf{u}_{\mathrm{s}}^{\mathsf H}\big(\mathbf{R}_{\mathbf{\hat g},k}-\gamma_{\mathrm{s,th}}\mathbf{R}_{\mathrm{err},k}\big)\mathbf{u}_{\mathrm{s}}
	\geq0.
\end{align}
Therefore a convex equivalent of $\left(\mathcal{P}_3 \right) $ with fixed $\mathbf{p}_{\mathrm{p}}$ can be given by 
\begin{align}
	\label{P5}
	\boxed{\left(\mathcal{P}_5 \right) \max_{\mathbf{p}_{\mathrm{d}}\geq\mathbf{0}}\ \big\{ {\bar R}_{\text{sum}}^{\text{lb}}\left( \mathbf{p}_\text{d};\mathbf{p}_\text{d}^{(l-1)} \right)\big|\eqref{P3b}, (\ref{eq:Rkpd-DCP}),(\ref{eq:sensingSINR-DCP}) \big\}.}
\end{align}
Note that $(\mathcal P 5)$ is a convex problem which can be solved efficiently using off-the-shelf convex solvers, \emph{e.g.}, CVX. A detailed routine for obtaining the optimal pilot and data power allocation is outlined in \textbf{Algorithm~\ref{algorithm-AO-P}}. 

\begin{algorithm}[tb]
	\caption{The proposed AO algorithm for optimal power allocation.}
	\label{algorithm-AO-P}
	
	\KwIn{$\mathbf{p}_\text{p}^{(0)}$, $\mathbf{p}_\text{d}^{(0)}$,
		$\xi$,
		$\delta$, $\eta$}
	
	\KwOut{$\mathbf{p}_\text{p}^{\text{opt}}$, $\mathbf{p}_\text{d}^{\text{opt}}$}
	
	$l \leftarrow 1$\;
	
	\Repeat{convergence}{
		
		\tcc{\textcolor{brown}{Optimal $\mathbf{p}_{\text{p}}$}}
		$\imath \leftarrow 1$\;
		
		\Repeat{convergence}{
			
			\Repeat{convergence}{
				
				Obtain $\nabla_{\mathbf{p}_{\text{p}}}\mathcal{L}_{\bm{\upomega},\xi}\big(\mathbf{p}_{\text{p}}^{(\imath-1)},\bm{\uptau}^{(\imath-1)}\big)$
				via \textbf{Theorem~\ref{thm:grad_Pd}}\;
				
				\tcc{\textcolor{brown}{Update $\mathbf{p}_{\text{p}}$ and $\bm{\uptau}$}}
				Obtain $\mathbf{p}_{\text{p}}^{(\imath)}$ via~\eqref{eq:Update_Pd}\;
				$\mathbf{p}_{\text{p}}^{(\imath)} \leftarrow  \min \left\{\mathbf{p}_{\text{p}}^{(\imath)},\frac{PT-{{T}_{\text{d}}}{\mathbf{p}_{\text{d},k}^{(\imath-1)}}}{{{T}_{\text{p}}}}\right\}$\;
				Obtain $\tau_{k}^{(\imath)}, \forall k\in\mathcal{K}$ and $\tau_{\mathrm{s}}^{(\imath)}$ via~\eqref{eq:tau_c} and~\eqref{eq:tau_s}, respectively\;
			}
			\tcc{\textcolor{brown}{Update $\bm{\upomega}$ and $\xi$}}
			Update $\omega_{k}, \forall k\in\mathcal{K}$ and $\omega_{\mathrm{s}}$ via~\eqref{eq:omega_c} and~\eqref{eq:omega_s}, respectively\;
			Update $\xi$ via~\eqref{eq:xi}\;
			$\imath\leftarrow\imath+1$\;
		}
		$\mathbf{p}_\text{p}^{(l)}\leftarrow \mathbf{p}_\text{p}^{(\imath)}$\;
		
		\tcc{\textcolor{brown}{Optimal $\mathbf{p}_{\text{d}}$}}
		Obtain $\mathbf{p}_{\mathrm{d}}^{(l)}$ by solving $\left(\mathcal{P}_5 \right) $ for given $\mathbf{p}_{\mathrm{d}}^{(l-1)}$ and $\mathbf{p}_{\mathrm{p}}^{(l)}$\;
		$l \leftarrow l+1$\;
	} 
	
	$\mathbf{p}_\text{p}^{\text{opt}}\leftarrow \mathbf{p}_\text{p}^{(l)}$, $\mathbf{p}_\text{d}^{\text{opt}}\leftarrow \mathbf{p}_\text{d}^{(l)}$
\end{algorithm}

\subsection{Optimizing Receive Beamforming After Transmitting Pilots}
Once the optimal power allocation is obtained, assuming the MRC-based receiver combining, pilots are transmitted to perform channel estimation using an MMSE estimator. Following this, we derive closed-form expressions for the instantaneous SINR for the communication users and sensing. Based on the derived expressions, we formulate the problem of maximizing the instantaneous achievable sum rate, while guaranteeing the communication and sensing QoS.
The post-combining received signal at the BS for user-$k$ can be rewritten as
\begin{align}
	\label{eq:yd,k 2}
	{{\mathbf{y}}_{\text{d},k}}
	=\mathbf{u}_{k}^{\mathsf H}\sum\nolimits_{i\in \mathcal{K}}{\sqrt{{{P}_{\text{d},i}}}\left( \mathbf{\hat{z}}_{k}+{\bm{\upepsilon }_{i}}+{{\bm{\upvarepsilon }}_{i}} \right)\mathbf{x}_{\text{d},i}^{\mathsf T}}+\mathbf{u}_{k}^{\mathsf H}{{\mathbf{N}}_{\text{d}}},
\end{align}
where $\mathbf{\hat{z}}_{k}=\hat{\mathbf{h}}_k+{{{\mathbf{\hat{g}}}}_{k}},\ \forall k\in\mathcal{K}$.
Then, the instantaneous SINR for decoding $\mathbf{x}_{\text{d},k}$ at the BS can be expressed as
\begin{align}
	{{\gamma }_{k}} \! =\! \frac{{{P}_{\text{d},k}}\!\left| \mathbf{u}_{k}^{\mathsf H}{{{\mathbf{\hat{z}}}}_{k}} \right|^{2}}{\!\sum\nolimits_{m\in {\mathcal{K}}\!\setminus\! \{k\}}\!{{{P}_{\text{d},m}}\!\left|\! \mathbf{u}_{k}^{\mathsf H}{{{\mathbf{\hat{z}}}}_{m}} \!\right|\!^{2}}\!\!+\!\!\sum\nolimits_{i\in \mathcal{K}}\!{{{P}_{\text{d},i}}\mathbf{u}_{k}^{\mathsf H}{{\mathbf{R}}_{\text{err},i}}{{\mathbf{u}}_{k}}}\!\!+\!\!\sigma _k^{2}}.
\end{align}
Then, the instantaneous achievable rate of user-$k$ is
$ R_{k} =\frac{{{T}_{\text{d}}}}{T}\log_2\left( 1+{{\gamma }_{k}} \right)$ and the instantaneous sum rate of $K$ communication users is $ R_{\text{sum}} =\sum\nolimits_{k\in \mathcal{K}} R_{k}$.
Similarly, the instantaneous sensing SINR can be written as 
\begin{align}
	{{\gamma }_{\text{s}}}\!=\!\frac{\sum\nolimits_{k\in \mathcal{K}}\!{\left( {{T}_{\text{p}}}{{P}_{\text{p},k}}\!+\!{{T}_{\text{d}}}{{P}_{\text{d},k}} \right)\left| \mathbf{u}_{\text{s}}^{\mathsf H}{{{\mathbf{\hat{g}}}}_{k}} \right|^{2}}}{\sum\nolimits_{k\in \mathcal{K}}\!{\left( {{T}_{\text{p}}}{{P}_{\text{p},k}}\!+\!{{T}_{\text{d}}}{{P}_{\text{d},k}} \right)\mathbf{u}_{\text{s}}^{\mathsf H}{{\mathbf{R}}_{\text{err},k}}\mathbf{u}_{\text{s}}}\!+\!T \left\|\mathbf{u}_{\text{s}} \right\|_2^2
		\sigma^2}.
	\label{eq:gammas 3}
\end{align}

The optimization problem is then formulated as
\begin{subequations}
	\label{P6}
	\begin{align}
		\left(\mathcal{P}_{6}\right)  \max_{\mathbf{U}}\quad & R_{\text{sum}}\left( \mathbf{U} \right) \label{P6a}\\ 
		\text{s.t.} \quad & 
		R_{k}\left( \mathbf{U} \right)\ge {{R}_{\text{th},k}},\forall k\in \mathcal{K}, \label{P6b}\\
		& \gamma_\text{s}\left( \mathbf{U} \right)\ge \gamma _{\text{s,th}}, \label{P6c}\\ 
		& \left\|\mathbf{u}_{\text{s}} \right\|_2^2=1,\left\|\mathbf{u}_{k} \right\|_2^2=1,\forall k\in \mathcal{K}. \label{P6d} 
	\end{align}
\end{subequations}
We first note that the objective is decoupled with $\mathbf{u}_{\text{s}}$, and that 
$\gamma_\text{s}$ and $ \bar \gamma_\text{s}$ are highly similar in structure. Therefore, the optimal solution for $\mathbf{u}_{\text{s}}$ can be derived based on~\eqref{eq:us-opt}. After dropping the sensing QoS constraint, the problem can be rewritten as
\begin{align}
	\label{P7} 
	\boxed{\left(\mathcal{P}_7 \right) \max_{\mathbf{U}_{\text{c}}}\ \big\{~\eqref{P6a}\big|\eqref{P6b}, \left\|\mathbf{u}_{k} \right\|_2^2=1,\forall k\in \mathcal{K} \big\}.}
\end{align}
Define $\tilde{x}_k = \left(\mathbf{u}_{k}^{(n)}\right)^{\mathsf H}{{{\mathbf{\hat{z}}}}_{k}}$ and $\tilde{y}_k = \left\|\mathbf{u}_{k}^{(n)} \right\|_2^2 \sigma^2 + \sum\nolimits_{m\in \mathcal{K} \setminus \{k\}} {{{P}_{\text{d},m}} \left|\left(\mathbf{u}_{k}^{(n)}\right)^{\mathsf H}{{{\mathbf{\hat{z}}}}_{m}} \right|^{2}} + \sum\nolimits_{i\in \mathcal{K}} {{{P}_{\text{d},i}}\left(\mathbf{u}_{k}^{(n)}\right)^{\mathsf H}{{\mathbf{R}}_{\text{err},i}}{{\mathbf{u}}_{k}^{(n)}}} $, then following~\cite{25-WCNC-RIS-MISO}, a \emph{concave quadratic minorant} of $R_k(\mathbf{u}_k)$ is given by
\begin{align}
	&R_k(\mathbf{u}_k) \geq  \frac{T_{\mathrm{d}}}{T\ln(2)} \left[\ln \left(1 + \frac{P_{\text{d},k} |\tilde{x}_k|^2}{\tilde{y}_k }\right) - \frac{P_{\text{d},k} |\tilde{x}_k|^2}{\tilde{y}_k } 
	\right.
	\notag \\ &
	\left.\!-\! a_k \sum\nolimits_{m\in {\mathcal{K}}}\!{{{P}_{\text{d},m}}\!\left| \mathbf{u}_{k}^{\mathsf H}{{{\mathbf{\hat{z}}}}_{m}} \right|^{2}} - a_k  \sum\nolimits_{i\in \mathcal{K}}\!{{{P}_{\text{d},i}}\mathbf{u}_{k}^{\mathsf H}{{\mathbf{R}}_{\text{err},i}}{{\mathbf{u}}_{k}}}\!
	\right.
	\notag \\ &
	\left. \!-\!a_k \left\|\mathbf{u}_{k} \right\|_2^2 \sigma^2 \!+\! \frac{ P_{\text{d},k}}{\tilde{y}_k }\left(\tilde{x}_k^{\mathsf H}\mathbf{u}_{k}^{\mathsf H}\mathbf{\hat{z}}_k \!+\! \tilde{x}_k \mathbf{\hat{z}}_k^{\mathsf H} \mathbf{u}_{k} \right) \right] \!\triangleq\!  \tilde{R}_k(\mathbf{u}_k),
\end{align}
where $a_k = \frac{P_{\text{d},k} |\tilde{x}_k|^2}{\tilde{y}_k \left(P_{\text{d},k} |\tilde{x}_k|^2+\tilde{y}_k\right)}$.
Hence, a convex equivalent of~$\left(\mathcal{P}_7 \right) $ can be written as
\begin{subequations}
	\label{P8}
	\begin{align}
		\left(\mathcal{P}_{8}\right) \max_{\mathbf{U}_{\text{c}}}\  & \sum\nolimits_{k\in \mathcal{K}} 
		\left( \tilde{R}_k(\mathbf{u}_k) - \wp  \left\|\mathbf{u}_{k} \right\|_2^2  \right)\label{P8a}\\ 
		  \text{s.t.} \   &
		\tilde{R}_k(\mathbf{u}_k) \ge {{R}_{\text{th},k}},\forall k\in \mathcal{K}, \label{P8b}\\
		& 
		\left\|\mathbf{u}_{k} \right\|_2^2 \leq 1,\forall k\in \mathcal{K}, \label{P8c} 
	\end{align}
\end{subequations}
where we have relaxed the equality constraints in~\eqref{P7} with inequalities in~\eqref{P8c}. Moreover, to make sure that~\eqref{P8c} is binding at convergence, we have regularized the objective with $ \wp > 0$ being the regularization parameter~\cite{23-WCL-RIS}. 
One can note that $\left(\mathcal{P}_8 \right) $ is a convex problem and can be solved using an off-the-shelf solver, \emph{e.g.} CVX.
Finally, the proposed algorithm for optimizing the receive beamforming matrix $\mathbf{U}_{\text{c}}$ is summarized in \textbf{Algorithm~\ref{algorithm-SCA-U}}.

\begin{algorithm}[tb]
	\caption{The proposed SCA method for optimal receive beamforming.}
	\label{algorithm-SCA-U}
	
	\KwIn{$\mathbf{U}_{\text{c}}^{(0)}$, $\wp$}
	
	\KwOut{$\mathbf{U}_{\text{c}}^{\text{opt}}$}
	$n \leftarrow 0$\;
	\Repeat{convergence}{
		Obtain $\mathbf{U}_{\text{c}}^{(n+1)}$ by solving $\left(\mathcal{P}_8 \right) $ \;
		$n \leftarrow n+1$\;
	} 
	$\mathbf{U}_{\text{c}}^{\text{opt}}\leftarrow \mathbf{U}_{\text{c}}^{(n)}$
\end{algorithm}

\subsection{Complexity Analysis}
We now analyze the computational complexity of \textbf{Algorithm~\ref{algorithm-AO-P}} and \textbf{Algorithm~\ref{algorithm-SCA-U}}.
The complexity of statistical-CSI-based receive beamforming design is $\mathcal{O}(N_{\text{b}}K + N_{\text{b}}^3)$.
Furthermore, the AO-based power allocation algorithm has a \emph{per-iteration} complexity of $\mathcal{O} \left(K^2\log(1/\epsilon )/\epsilon  + K^{3.5}\right) $, where $\epsilon$ denotes the convergence tolerance.
Under instantaneous CSI, the proposed SCA-based algorithm for optimizing the receive beamforming vectors of the communication users has a \emph{per-iteration} complexity of $\mathcal{O}\left(N_{\text{b}}^{1.5}K^{2.5} + N_{\text{b}}^3K^3 \right)$.
In addition, optimizing the receive beamforming vector of the sensing target requires $\mathcal{O}(N_{\text{b}}^3)$ operations.
Since in practical systems it typically holds that $N_{\text{b}} > K$, the overall \emph{per-iteration} complexity of the proposed optimization framework can be simplified to $\mathcal{O} \left(N_{\text{b}}^3K^3 \right)$.
Moreover, the complexities of the benchmark schemes employing ZF and MRC as the receive beamforming strategies scale as $ \mathcal{O}(K^3+K^2 N_{\text{b}})$ and $ \mathcal{O}(1)$, respectively. Consequently, for a fixed $\epsilon$, the overall \emph{per-iteration} complexity of both the ZF- and MRC-based benchmark schemes is dominated by the optimal power allocation and can be approximated as $\mathcal{O} \left(K^{3.5} \right) $.

\section{Simulation Results}
In this section, comprehensive numerical results are presented to evaluate the performance of the proposed ISAC system.
To facilitate channel estimation and reduce pilot contamination, the pilot sequences are selected from the columns of a unitary discrete Fourier transform (DFT) matrix.
Besides, the data symbols are assumed to follow a standard complex Gaussian distribution, which is commonly used for rate analysis.
In addition, the large-scale path loss for user-$k$ to the BS links is modeled as $ L_{\mathrm{h},k}=\left(d_{\text{u2b},k}/d_{0}\right)^{-\alpha_{\text{u2b}}}$, where $d_{\text{u2b},k}$ denotes the distance between user-$k$ and the BS, $d_{0}$ denotes the reference distance, and $\alpha_{\text{u2b}}$ denotes the path loss exponent for the communication channel between the user and the BS. 
The covariance matrix of $\mathbf{g}_{k}$ as $\mathbf{R}_{\mathrm{g},k}=\alpha_{\text{RCS}}L_{\mathrm{g},k}\sigma_{\mathrm{g},k}^{2}\mathbf{a}\left(\theta_\text{t}\right)\mathbf{a}^{\mathsf H}\left(\theta_\text{t}\right)$, where $L_{\mathrm{g},k}=\left(d_{\text{u2t},k}/d_{0}\right)^{-\alpha_{\text{u2t}}}\left(d_{\text{t2b}}/d_{0}\right)^{-\alpha_{\text{t2b}}}$ is the large-scale path loss for user-$k$ to target then to the BS. $d_{\text{u2t},k}$ and $d_{\text{t2b}}$ represent the distances from user-$k$ to the target and the target to the BS, respectively. $\alpha_{\text{u2t}}$ and $\alpha_{\text{t2b}}$ represent the path loss exponent for the two paths, respectively.
The BS is at the origin $(0,0)$~m and the target is at $\left(d_{\text{t2b}}/\sqrt{2},\ d_{\text{t2b}}/\sqrt{2}\right)$~m. 
We assume that $K$ communication users are randomly distributed within a circular area of radius $R$ centered at $(100,0)$~m. 
Unless otherwise stated, we set
$ K = 4 $, $ N_{\text{b}} = 6 $, $P = 10~\text{dBm}$,
$T = 100$, $T_{\text{p}} = K+1$,
$\sigma^{2} = -70 \ \text{dBm/Hz}$,
$P_{\text{FA}} = 10^{-5}$,
$\alpha_{\text{u2b}} = 3.6$,
$\alpha_{\text{u2t}} = \alpha_{\text{t2b}} = 2.2$,
$\alpha_{\text{RCS}} = 0.8$,
$d_0 = 1 \ \text{m}$,
$d_{\text{t2b}} = 50 \ \text{m}$,
$R  = 100 \ \text{m} $,
$\xi = 1$,
$\delta = 1$, $\eta = 1.5$,
$ \wp  = 10^{-2}$,
$ R_{\text{th},k} = 1 \ \text{bps/Hz}, \forall k\in \mathcal{K}$, $P_\text{D,th} = 0.99 $, and $\epsilon = 10^{-3}$.

\subsection{Algorithm Performance Evaluation}

\begin{figure}[t] 
	\centering
	\includegraphics[width=0.355\textwidth]
	{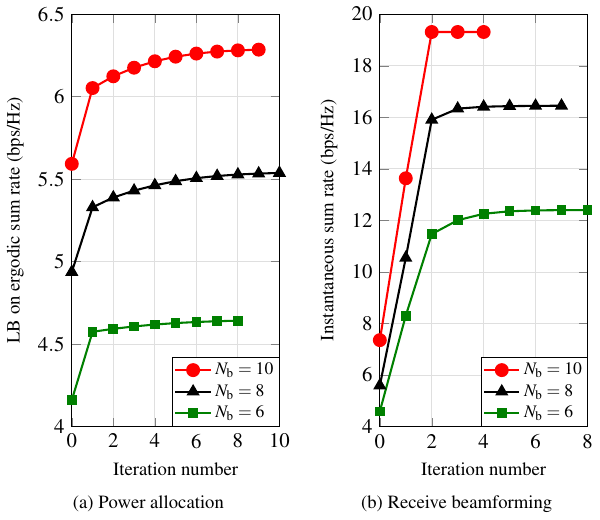}		
	\caption{Convergence results for the power allocation (\textbf{Algorithm~\ref{algorithm-AO-P}}) on the left and the receive beamforming (\textbf{Algorithm~\ref{algorithm-SCA-U}}) on the right.} 
	\label{Fig_Convergence}
\end{figure}
Figure~\ref{Fig_Convergence} illustrates the convergence behavior of the proposed power allocation (\textbf{Algorithm~\ref{algorithm-AO-P}}) and receive beamforming (\textbf{Algorithm~\ref{algorithm-SCA-U}}) algorithms for different values of $N_\text{b}$. Both algorithms produce non-decreasing iterates and converge within a few iterations, confirming their fast convergence and numerical stability.
Note that the initial instantaneous sum rate in Fig.~\ref{Fig_Convergence}(b) may be lower than the converged value in Fig.~\ref{Fig_Convergence}(a). This is because \textbf{Algorithm~\ref{algorithm-AO-P}} (in Fig.~\ref{Fig_Convergence}(a)) optimizes a lower bound on the \emph{ergodic} sum rate (based on statistical CSI), while \textbf{Algorithm~\ref{algorithm-SCA-U}} (in Fig.~\ref{Fig_Convergence}(b)) maximizes the \emph{instantaneous} sum rate (based on instantaneous CSI).
The observed improvement in the instantaneous sum rate emphasizes the benefit of optimizing receive beamforming after power allocation and validates the effectiveness of the joint design in leveraging the estimated instantaneous channel to enhance overall system performance.

\begin{figure}[t] 
	\centering
	\includegraphics[width=0.8\columnwidth]{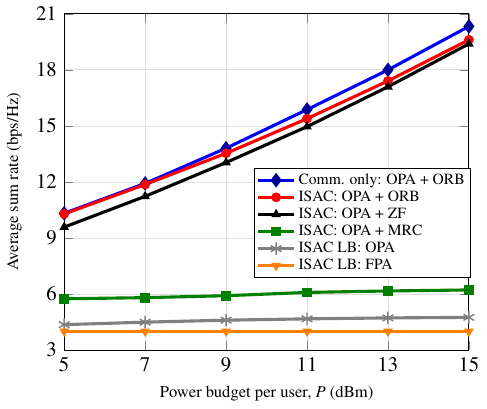}			
	\caption{Impact of the power budget per user on the average sum rate.}
	\label{Fig_P_compare}
\end{figure}
 
In Fig.~\ref{Fig_P_compare}, we present the impact of the power budget per user $P$ on the average sum rate. 
First, we evaluate the effectiveness of the proposed power allocation algorithm. 
The ``ISAC LB: FPA'' serves as a baseline that employs a fixed feasible power allocation (FPA), whereas ``ISAC LB: OPA'' represents the lower-bound performance achieved by the proposed optimal power allocation (OPA) based on statistical CSI. 
This comparison clearly demonstrates the benefit of OPA. For instance, at $P = 15$~dBm, OPA achieves an 18.82\% performance gain over the fixed scheme.
Next, we assess the average sum rate over 100 channel realizations for different receive beamforming strategies, all based on instantaneous CSI with OPA.
The ``ISAC: OPA + MRC'' represents an MRC receiver whose performance is fundamentally limited, as it only maximizes the desired signal component without effectively mitigating multi-user interference (MUI) and target-echo interference.
In contrast, ``ISAC: OPA + ZF'' employs a ZF receiver, which effectively eliminates MUI by nullifying the interference subspace and thereby achieves a significant performance improvement.
Moreover, ``ISAC: OPA + ORB'' depicts the average sum rate achieved by the proposed optimal receive beamforming (ORB) design obtained via \textbf{Algorithm~\ref{algorithm-SCA-U}}.

A key observation is that ORB consistently outperforms ZF, with the performance gain being most significant at lower power budgets.
This result is critical, as it highlights the fundamental trade-off: in the noise-limited regime (low $P$), the ZF receiver's strict interference-nulling strategy (which is unnecessary when interference is weak) leads to severe noise enhancement and degraded performance.
In contrast, the proposed ORB adaptively balances the interference-noise trade-off, tolerating minor interference to avoid amplifying dominant noise and thereby achieving a higher sum rate.
As $P$ increases, the system becomes interference-limited, and both ORB and ZF gradually converge to an interference-nulling strategy, naturally narrowing the performance gap. 
Finally, to quantify the impact of sensing integration, we also include the communication-only system (\emph{i.e.}, without target sensing), denoted as ``Comm. only: OPA + ORB''. 
This confirms that integrating sensing into communication incurs only a negligible rate reduction. 
For example, at $P = 15$~dBm, the ISAC system experiences merely a 3.54\% performance loss compared with the communication-only counterpart, thanks to the efficient power allocation and receive beamforming design.

\subsection{Communication Performance Analysis}

\begin{figure}[t] 
	\centering
	\includegraphics[width=0.355\textwidth]{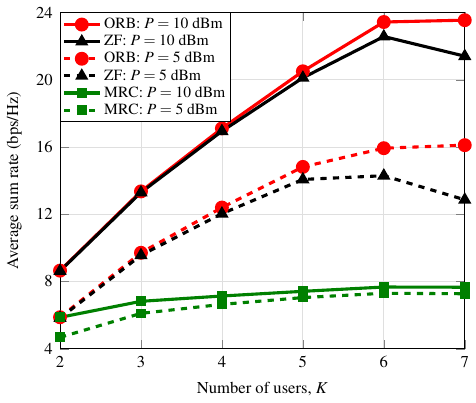}			
	\caption{Impact of the number of users on the average sum rate for $N_{\text{b}} = 8$.}
	\label{Fig_K}
\end{figure}

Figure~\ref{Fig_K} shows the impact of the number of users $K$ on the average sum rate for the ORB-, ZF-, and MRC-based schemes with $N_{\text{b}} = 8$. 
It is evident that the MRC-based scheme achieves the lowest sum rate due to severe MUI and target-echo interference. 
Although its performance slightly improves for small $K$, the sum rate quickly saturates as $K$ increases and interference becomes dominant, while increasing transmit power provides only marginal benefits. 
The key insight from this figure lies in the changing gap between the ZF and proposed ORB schemes. 
In the low-load regime (\emph{e.g.}, $K = 2 \text{ to } 4$), ORB achieves nearly identical performance to ZF, confirming the correctness of the proposed algorithm. 
As $K$ further increases, the ZF performance deteriorates sharply and eventually collapses when $K$ approaches $N_{\text{b}}$, due to severe noise enhancement caused by its strict interference-nulling strategy. 
In contrast, the proposed ORB remains stable and robust across all loading conditions. 
In the high-load region where ZF fails, ORB adaptively balances interference suppression and noise amplification, retaining minor residual interference to avoid excessive noise enhancement and thereby achieving a much higher sum rate. 
This demonstrates that ORB effectively mitigates the noise amplification issue inherent in ZF and provides superior robustness in dense multi-user ISAC systems.

\begin{figure}[t] 
	\centering
	\includegraphics[width=0.355\textwidth]{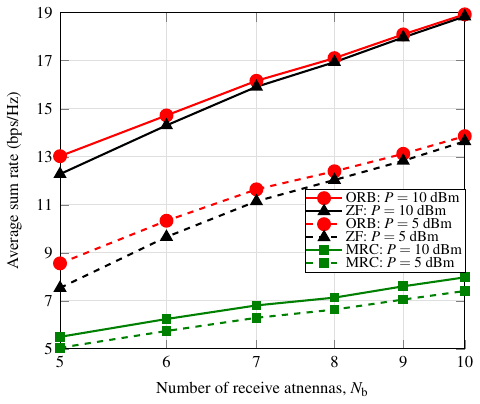}	  
    \caption{Impact of the number of receive antennas on the average sum rate.}
	\label{Fig_R_Nb}
\end{figure}

In Fig.~\ref{Fig_R_Nb}, we depict the impact of the number of receive antennas at the BS $N_{\text{b}}$ on the average sum rate. As $N_{\mathrm{b}}$ increases, the achievable sum rate improves due to the \emph{receive beamforming gain}. 
Consistent with the trend observed in Fig.~\ref{Fig_K}, the MRC scheme again exhibits the poorest performance. 
For the ZF and ORB schemes, the system becomes resource-constrained at an 80\% load ($K=4$, $N_{\mathrm{b}}=5$), representing the most challenging scenario for ZF, as also shown in Fig.~\ref{Fig_K}. 
For example, at $P=5$~dBm and $N_{\mathrm{b}}=5$, the ZF scheme achieves a sum rate of 7.54~bps/Hz, while the proposed ORB scheme attains 8.56~bps/Hz, corresponding to a 13.53\% improvement. 
As $N_{\mathrm{b}}$ increases, the system gains additional spatial degrees of freedom, and the ZF performance rapidly approaches that of ORB. 
Together with Fig.~\ref{Fig_K}, these results demonstrate that the proposed ORB maintains superior robustness under resource-limited conditions. 
Conversely, in resource-rich scenarios, ORB naturally converges to the near-optimal ZF solution, which further validates its robustness and consistency.

\begin{figure}[t] 
	\centering
	\begin{subfigure}[t]{0.355\textwidth}
		\centering
		\includegraphics[width=\textwidth]{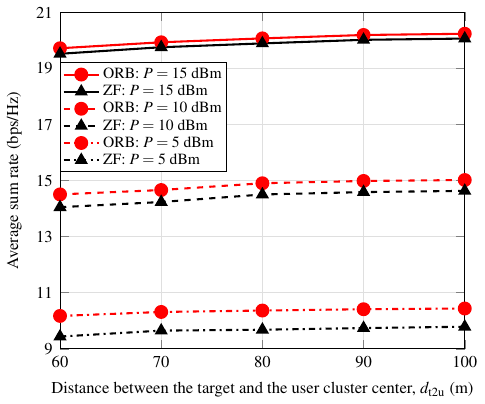}
		\caption{Sum rate}
		\label{Fig_D_R}
	\end{subfigure}
	\hfill
	\begin{subfigure}[t]{0.355\textwidth}
		\centering
		\includegraphics[width=\textwidth]{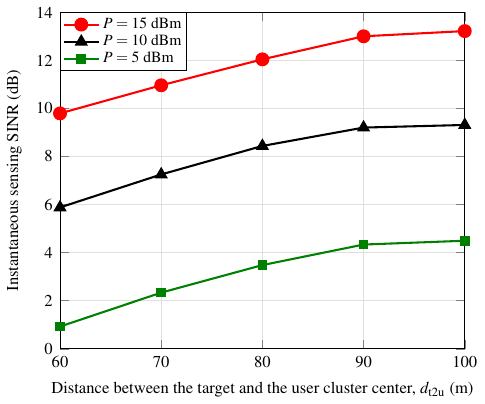}
		\caption{Sensing SINR}
		\label{Fig_D_SINR}
	\end{subfigure}
	\caption{Impact of the distance between the target and the user cluster center on the average sum rate and the sensing SINR.}
	\label{Fig_D_SINR_R}
\end{figure}

Figure~\ref{Fig_D_SINR_R} illustrates the impact of the distance between the target and the user cluster center, denoted as $d_{\mathrm{t2u}}$, on the average sum rate and the sensing SINR for the ORB- and ZF-based schemes. In this setup, the target moves along the $x$-axis between the BS and the center of the user cluster, such that $d_{\mathrm{t2u}}$ increases while $d_{\mathrm{t2b}}$ decreases.
As shown in Fig.~\ref{Fig_D_SINR_R}, both the sum rate (in Fig.~\ref{Fig_D_SINR_R}(a)) and the sensing SINR (in Fig.~\ref{Fig_D_SINR_R}(b)) improve as $d_{\mathrm{t2u}}$ increases, although their growth scales differ significantly.
The sensing SINR exhibits a pronounced increase, since in the bistatic sensing configuration the performance is predominantly governed by the target-to-BS distance $d_{\mathrm{t2b}}$. Although a larger $d_{\mathrm{t2u}}$ introduces higher path loss in the first hop (user-to-target), the substantial reduction in $d_{\mathrm{t2b}}$ markedly enhances the echo power received at the BS, resulting in a significant SINR gain.
In contrast, the sum rate increases much more gradually. This is because the communication performance remains dominated by the direct user--BS links, while the reflected user--target--BS path provides only a secondary contribution. As a result, the approaching target leads to a slow but consistent enhancement of the effective uplink channel conditions and a modest increase in the achievable sum rate.

\begin{figure}[t] 
	\centering
	\begin{subfigure}[t]{0.355\textwidth}
		\centering
		\includegraphics[width=\textwidth]{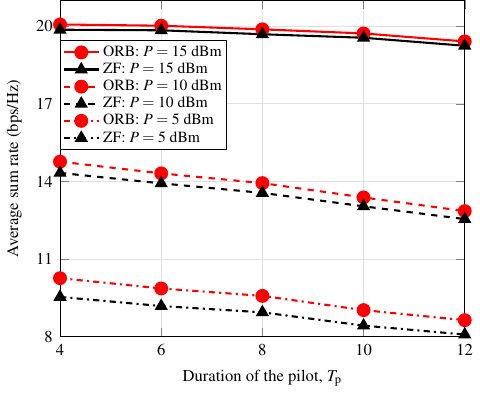}
		\caption{Sum rate}
		\label{Fig_Tp}
	\end{subfigure}
	\hfill
	\begin{subfigure}[t]{0.355\textwidth}
		\centering
		\includegraphics[width=\textwidth]{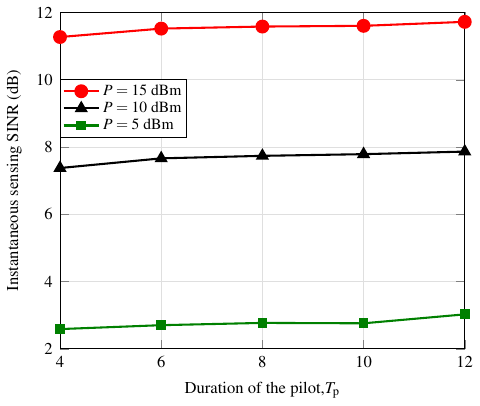}
		\caption{Sensing SINR}
		\label{Fig_Tp_SINR}
	\end{subfigure}
	\caption{Impact of the pilot duration on the sum rate and the sensing SINR.}
	\label{Fig_Tp_R_SINR}
\end{figure}

In Fig.~\ref{Fig_Tp_R_SINR}, we present the impact of the pilot duration $T_{\mathrm{p}}$ on the sum rate and sensing SINR, highlighting the inherent trade-off between communication and sensing. As shown in Fig.~\ref{Fig_Tp_R_SINR}~(b), the sensing SINR monotonically increases with $T_{\mathrm{p}}$ because longer pilots improve channel estimation accuracy, allowing the BS to suppress dominant user interference and detect the weak echo more reliably. Conversely, Fig.~\ref{Fig_Tp_R_SINR}~(a) shows that the average sum rate decreases with $T_{\mathrm{p}}$ due to the reduced data transmission time, $(T-T_p)/T$, which more than offsets the minor improvement from more accurate channel estimates. This figure thus clearly illustrates the conflict: Sensing performance favors longer pilots for more accurate channel estimation, whereas communication performance prefers shorter pilots to maximize data throughput.

\subsection{Sensing Performance Analysis}

\begin{figure}[t] 
	\centering
	\includegraphics[width=0.355\textwidth]{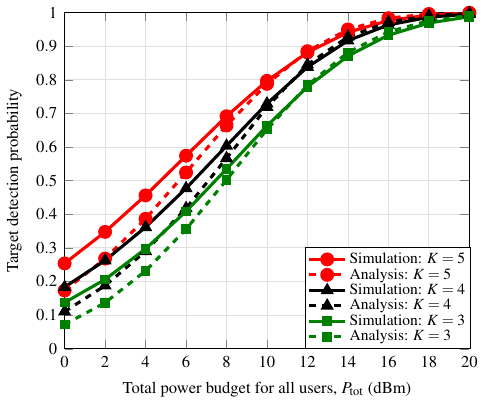}		
	\caption{Impact of the total power budget for all communication users on the target detection probability at $d_{\text{t2b}}=100\text{m}$.}
	\label{Fig_Ptot}
\end{figure}  

We further compare the theoretically derived target detection probability with its simulated counterpart, averaged over 100,000 Monte Carlo channel realizations. In this setup, the pilot and data transmission energy are equally allocated, with communication receive beamforming vectors obtained via MRC and the sensing target beamforming vector via EVD.
To highlight the effect of spatial diversity from varying $K$, the total transmit power is fixed at $P_\text{tot}$, and each user is assigned an equal share $P = P_\text{tot}/K$.
Fig.~\ref{Fig_Ptot} illustrates the impact of $P_\text{tot}$ on the target detection probability at $d_{\text{t2b}}=100$~m. The results indicate that the theoretical detection probability closely matches the simulated values, validating the analytical expressions. Moreover, increasing $K$ significantly improves detection performance, demonstrating that the additional spatial diversity from more users effectively enhances sensing under a fixed total transmit power. This substantial improvement highlights the benefit of multi-user cooperation in the proposed detection scheme.

\section{Conclusions}
In this paper, we have proposed a multi-user collaborative uplink ISAC system. 
A closed-form lower bound on the ergodic sum rate, as well as an expression for the probability of detection, have been derived. Furthermore, a two-stage optimization framework has been proposed to maximize the sum rate by jointly optimizing pilot and data power allocation and receive beamforming. We optimized the pilot power allocation using a PDD-based gradient ascent method, while data power allocation and receive beamforming were handled via SCA.   
Simulation results demonstrated the benefits and superiority of the proposed optimization framework and collaborative sensing in the uplink ISAC networks.

\begin{appendices}
	\section{\label{sec:Proof-bar-gamma-k}Proof of Theorem~\ref{thm:bar gamma k}}
	If $\mathbf{u}_{k}= \frac{\mathbf{\hat z}_{k}}{\sqrt{\mathbb{E} \left\{||\mathbf{\hat z}_{k}||_2^2\right\}}},\ \forall k\in\mathcal{K}$ is used, we can obtain that
	\begin{align}
        \left|\mathbb{E}\left\{ \mathbf{u}_{k}^{\mathsf H}{{\mathbf{z}}_{k}} \right\}\right|^2
        =\text{tr}\left( {{\mathbf{R}}_{\text{est},k}} \right), \  
		\mathbb{E}\left\{ \left\| {{\mathbf{u}}_{k}} \right\|_{2}^{2} \right\}=1.
        \label{eq:u_k z_k 2}
	\end{align}
    For $\mathbb{E}\left\{ \left| \mathbf{u}_{k}^{\mathsf H}{{\mathbf{z}}_{i}} \right|^{2} \right\}$, we can obtain that 
    	\begin{align}
		\mathbb{E}\left\{ {{\left| \mathbf{u}_{k}^{\mathsf H}{{\mathbf{z}}_{i}} \right|}^{2}} \right\}\!= \!
		\begin{cases}
			\frac{\text{tr}\left( {{\mathbf{R}}_{i}}{{\mathbf{R}}_{\text{est},k}} \right)}{\text{tr}\left({\mathbf{R}}_{\text{est},k}\right)}, i\!\ne \!k, \\ 
			\frac{\mathbb{E}\left\{ \!{{\left| \mathbf{\hat{z}}_{k}^{\mathsf H}{{{\mathbf{\hat{z}}}}_{k}} \right|}^{2}} \!\right\}+\text{tr}\left( \left({{\mathbf{R}}_{i} - {\mathbf{R}}_{\text{est},k}}\right){{\mathbf{R}}_{\text{est},k}} \right)}{\text{tr}\left({\mathbf{R}}_{\text{est},k}\right)}, i\!=\!k, 
		\end{cases}
        \label{eq:u_k z_hat_k}
	\end{align}
    where ${{\mathbf{R}}_{i}}={{\mathbf{R}}_{\text{h},i}}+{{\mathbf{R}}_{\mathrm{g},i}}$.
	Then, based on
	${{\mathbf{\hat{z}}}_{k}}=\sqrt{{{P}_{\text{p},k}}}{{\mathbf{R}}_{k}}{{\left( {{P}_{\text{p},k}}{{T}_{\text{p}}}{{\mathbf{R}}_{k}}+\sigma^2\mathbf{I} \right)}^{-1}}{{\mathbf{y}}_{\text{p},k}}$,  $\mathbb{E}\left\{ {{\left| \mathbf{\hat{z}}_{k}^{\mathsf H}{{{\mathbf{\hat{z}}}}_{k}} \right|}^{2}} \right\}$ can be rewritten as
	\begin{align}
		\label{eq:z k hat z k hat}
		\mathbb{E}\left\{ {{\left| \mathbf{\hat{z}}_{k}^{\mathsf H}{{{\mathbf{\hat{z}}}}_{k}} \right|}^{2}} \right\}=P_{\text{p},k}^{2}\mathbb{E}\left\{ {{\left| \mathbf{y}_{\text{p},k}^{\mathsf H}\mathbf{C}_{k}^{-1}{{\mathbf{R}}_{k}}{{\mathbf{R}}_{k}}\mathbf{C}_{k}^{-1}{{\mathbf{y}}_{\text{p},k}} \right|}^{2}} \right\},
	\end{align}
	where ${{\mathbf{C}}_{k}}={{P}_{\text{p},k}}{{T}_{\text{p}}}{{\mathbf{R}}_{k}}+\sigma^2\mathbf{I}$.
	Based on~\eqref{eq:y_pk}, we can obtain  ${{\mathbf{y}}_{\text{p},k}}\sim \mathcal{C}\mathcal{N}\left( \mathbf{0},{{P}_{\text{p},k}}T_{\text{p}}^{2}{{\mathbf{R}}_{k}}+{{T}_{\text{p}}}\sigma^2\mathbf{I} \right)$.
	According to $\mathbb{E}\left\{ {{\left| {{\mathbf{a}}^{\mathsf H}}\mathbf{Ba} \right|}^{2}} \right\}={{\left| \text{tr}\left( \mathbf{BA} \right) \right|}^{2}}+\text{tr}\left( \mathbf{BA}{{\mathbf{B}}^{\mathsf H}}\mathbf{A} \right)$ with $\mathbf{a}\sim \mathcal{C}\mathcal{N}\left( 0,\mathbf{A} \right)$,~\eqref{eq:z k hat z k hat} can be expressed as
	\begin{align}
		\label{eq:z-hat-k 3}
		& \mathbb{E}\left\{ {{\left| \mathbf{\hat{z}}_{k}^{\mathsf H}{{{\mathbf{\hat{z}}}}_{k}} \right|}^{2}} \right\}
        ={{\left| \text{tr}\left( {{\mathbf{R}}_{\text{est},k}} \right) \right|}^{2}}+\text{tr}\left( {{\mathbf{R}}_{\text{est},k}}{{\mathbf{R}}_{\text{est},k}} \right).
	\end{align}	
	Substituting~\eqref{eq:z-hat-k 3} into~\eqref{eq:u_k z_hat_k}, we can obtain
	\begin{align}
		\label{eq:u_k z_i}
		\mathbb{E}\left\{ {{\left| \mathbf{u}_{k}^{\mathsf H}{{\mathbf{z}}_{i}} \right|}^{2}} \right\}= \!
		\begin{cases}
			\frac{\text{tr}\left( {{\mathbf{R}}_{i}}{{\mathbf{R}}_{\text{est},k}} \right)}{\text{tr}\left({\mathbf{R}}_{\text{est},k}\right)}, i\ne k, \\ 
			\frac{{{\left| \text{tr}\left( {{\mathbf{R}}_{\text{est},k}} \right) \right|}^{2}}+\text{tr}\left( {{\mathbf{R}}_{\text{est},k}}{{\mathbf{R}}_{\text{est},k}} \right)}{\text{tr}\left({\mathbf{R}}_{\text{est},k}\right)}, i=k. 
		\end{cases}
	\end{align}
	By substituting~\eqref{eq:u_k z_k 2} and~\eqref{eq:u_k z_i} into~\eqref{eq:bar gamma k}, ${{\bar{\gamma }}_{k}}$ in~\eqref{eq:bar gamma k 2} can be obtained. This completes the proof. 
	
	\section{\label{sec:Proof-grad_Pd}Proof of Theorem~\ref{thm:grad_Pd}}
	Based on~\eqref{eq:L}, for a given $\ell\in\mathcal{K}$, we have 
	\begin{align}
		\label{eq:grad_Pd 2}
		& \nabla_{P_{\mathrm{p},\ell}}\mathcal{L}_{\bm{\upomega},\xi}(\mathbf{p}_{\text{p}},\bm{\uptau})=\sum\nolimits_{k\in\mathcal{K}}\nabla_{P_{\mathrm{p},\ell}}{{\bar R}_{k}^{\text{lb}}}(\mathbf{p}_{\mathrm{p}})
		\notag \\ & 
		\ \ \ \ \ \ 
		-\sum\nolimits_{k\in\mathcal{K}}\left(\omega_{k}+\frac{1}{\xi}{f_k}\left(\mathbf{p}_{\text{p}},{{\tau }_{k}}\right)\right)\nabla_{P_{\mathrm{p},\ell}}f_{k}(\mathbf{p}_{\mathrm{p}},\tau_{k})
		\notag \\ &
		\ \ \ \ \ \ 
		-\left(\omega_{\mathrm{s}}+\frac{1}{\xi}{{f}_{\text{s}}}\left(\mathbf{p}_{\text{p}},{{\tau }_\text{s}}\right)\right)\nabla_{P_{\mathrm{p},\ell}}f_{\mathrm{s}}(\mathbf{p}_{\mathrm{p}},\tau_{\mathrm{s}}).
	\end{align}
	Then, $\nabla_{P_{\mathrm{p},\ell}}{{\bar R}_{k}^{\text{lb}}}(\mathbf{p}_{\mathrm{p}})$ can be written as
	\begin{align}
		\label{eq:grad_R_k}
		&\nabla_{P_{\mathrm{p},\ell}}{{\bar R}_{k}^{\text{lb}}}(\mathbf{p}_{\mathrm{p}})
        \!=\!\frac{{{T}_{\text{d}}}}{T\ln(2)}{{\nabla }_{P_{\mathrm{p},\ell}}}\ln \left(\! 1\!+\!\frac{{{P}_{\text{d},k}}\text{tr}\left( {{\mathbf{R}}_{\text{est},k}} \right)}{\sum\limits_{{i\in \mathcal{K}}}{{{P}_{\text{d},i}}\frac{\text{tr}\left( {{\mathbf{R}}_{i}}{{\mathbf{R}}_{\text{est},k}} \right)}{\text{tr}\left( {{\mathbf{R}}_{\text{est},k}} \right)}}\!+\!\sigma^2}\!\right)
		\notag \\ 
		&
        \! =\!\frac{{{T}_{\text{d}}}}{T\ln(2)}{{\nabla }_{P_{\mathrm{p},\ell}}}\left[
		\!-\!\ln\left(\sum\limits_{i\in\mathcal{K}}{{P}_{\text{d},i}}\text{tr}\left( {\mathbf{R}}_i {{\mathbf{R}}_{\text{est},k}} \right)\!+\!\sigma^2\text{tr}\left( {{\mathbf{R}}_{\text{est},k}} \right)\right)\right.
		\notag \\ 
		& 
		\left.\!+\!\ln\!\!\left(\!P_{\mathrm{d},k}\text{tr}^2\!\left(\! {{\mathbf{R}}_{\text{est},k}}\! \right)   \!+\!\sum\limits_{i\in\mathcal{K}}\!{{P}_{\text{d},i}}\text{tr}\!\left( \!{\mathbf{R}}_i {{\mathbf{R}}_{\text{est},k}} \!\right)\!+\!\sigma^2\text{tr}\!\left(\! {{\mathbf{R}}_{\text{est},k}} \!\right)\!\!\right)\!\!\right]\!.
	\end{align}
	In addition, $\nabla_{P_{\mathrm{p},\ell}}\mathbf{R}_{\mathrm{est},k} $ can be derived as
	\begin{align}
		\label{eq:grad_R_est}
		&\nabla_{P_{\mathrm{p},\ell}}\mathbf{R}_{\mathrm{est},k} = \nabla_{P_{\mathrm{p},\ell}} \left( {{P}_{\text{p},k}}{{T}_{\text{p}}}{{\mathbf{R}}_{k}}{\left( {{P}_{\text{p},k}}{{T}_{\text{p}}}{{\mathbf{R}}_{k}}+\sigma^2\mathbf{I} \right)^{-1}}{{\mathbf{R}}_{k}}\right)
		\notag \\ &
		\!= \!
		\begin{cases}
			T_{\mathrm{p}}\mathbf{R}_{\ell}\mathbf{C}_\ell^{-1}\mathbf{R}_{\ell}\!-\!P_{\mathrm{p},\ell}T_{\mathrm{p}}^2\mathbf{R}_{\ell}\mathbf{C}_\ell^{-1}\mathbf{R}_{\ell}\mathbf{C}_\ell^{-1}\mathbf{R}_{\ell}, \ell\!=\!k,
			\\
			0, 
			\ \ \ \ \ \ \ \ \ \ \ \ \ \ \ \ \ \ \ \ \ \ \ \ \ \ \ \ \ \ \ \ \ \ \ \ \ \ \ \ \ \ \ \ \ \ \ 
			\ell\! \neq \!k,
		\end{cases}
	\end{align}
	where
	$\mathbf{C}_\ell = {{P}_{\text{p},\ell}}{{T}_{\text{p}}}{{\mathbf{R}}_{\ell}}+\sigma^2\mathbf{I}$.
	According to~\eqref{eq:grad_R_k} and~\eqref{eq:grad_R_est}, the closed-form expression for 
	$\nabla_{P_{\mathrm{p},\ell}}{{\bar R}_{k}^{\text{lb}}}(\mathbf{p}_{\mathrm{p}})$
	can be obtained as
	\begin{align}
		\label{eq:grad_R_k 2}
		&\nabla_{P_{\mathrm{p},\ell}}{{\bar R}_{k}^{\text{lb}}}(\mathbf{p}_{\mathrm{p}}) = 
		\frac{{{T}_{\text{d}}}}{T\ln(2)}\left(-\frac{\sum\nolimits_{i\in\mathcal{K}}{{P}_{\text{d},i}}\text{tr}\left( {\mathbf{R}}_i \acute{\mathbf{R}_{\mathrm{est},k}} \right)}{t_{\text{D},k} }\right.
		\notag \\ &
		\ \ \ \ \ \ 
		\left.-\frac{\sigma^2\text{tr}\left( \acute{\mathbf{R}_{\mathrm{est},k}}\right)}{t_{\text{D},k} }+\frac{2P_{\mathrm{d},k}\text{tr}\left( {{\mathbf{R}}_{\text{est},k}} \right)\text{tr}\left( \acute{\mathbf{R}_{\mathrm{est},k}}\right)}{t_{\text{N},k} } \right.
		\notag \\ &
		\ \ \ \ \ \ 
		\left. +\frac{\sum\nolimits_{i\in\mathcal{K}}{{P}_{\text{d},i}}\text{tr}\left( {\mathbf{R}}_i \acute{\mathbf{R}_{\mathrm{est},k}} \right)}{t_{\text{N},k} }+\frac{\sigma^2\text{tr}\left( \acute{\mathbf{R}_{\mathrm{est},k}}\right)}{t_{\text{N},k} }\right),
	\end{align}
	where
	$\acute{\mathbf{R}_{\mathrm{est},k}} = \nabla_{P_{\mathrm{p},\ell}}\mathbf{R}_{\mathrm{est},k}$, 
	$t_{\text{D},k}  = \sum\nolimits_{i\in\mathcal{K}}{{P}_{\text{d},i}}\text{tr}\left( {\mathbf{R}}_i {{\mathbf{R}}_{\text{est},k}} \right)+\sigma^2\text{tr}\left( {{\mathbf{R}}_{\text{est},k}} \right)$,
	and 
	$t_{\text{N},k}  = P_{\mathrm{d},k}\text{tr}^2\left( {{\mathbf{R}}_{\text{est},k}} \right)   +\sum\nolimits_{i\in\mathcal{K}}{{P}_{\text{d},i}}\text{tr}\left( {\mathbf{R}}_i {{\mathbf{R}}_{\text{est},k}} \right)+\sigma^2\text{tr}\left( {{\mathbf{R}}_{\text{est},k}} \right)$.
	Similarly, we can obtain that
	\begin{align}
		\label{eq:grad_f_k}
		\nabla_{P_{\mathrm{p},\ell}}f_{k}(\mathbf{p}_{\mathrm{p}},\tau_{k}) = -\frac{1}{R_{\mathrm{th},k}}
		\nabla_{P_{\mathrm{p},\ell}}{{\bar R}_{k}^{\text{lb}}}(\mathbf{p}_{\mathrm{p}}).
	\end{align}
	At the end, 
	$\nabla_{P_{\mathrm{p},\ell}}f_{\mathrm{s}}(\mathbf{p}_{\mathrm{p}},\tau_{\mathrm{s}})$
	can be written as 
	\begin{align}
		\label{eq:grad_f_s}
		\nabla_{P_{\mathrm{p},\ell}}f_{\mathrm{s}}(\mathbf{p}_{\mathrm{p}},\tau_{\mathrm{s}}) 
		=\frac{{t_{\text{N,s}}\nabla_{P_{\mathrm{p},\ell}} t_{\text{D,s}}-t_{\text{D,s}}\nabla_{P_{\mathrm{p},\ell}} t_{\text{N,s}}}}{\gamma _{\text{s,th}} t_{\text{D,s}}^2},
	\end{align}
	where $t_{\text{N,s}}  = \sum\nolimits_{k\in \mathcal{K}}\left( {{T}_{\text{p}}}{{P}_{\text{p},k}}+{{T}_{\text{d}}}{{P}_{\text{d},k}} \right){{\mathbf{u}}_{\text{s}}^{\mathsf H}}{{\mathbf{R}}_{\hat{\mathrm{g}},k}}\mathbf{u}_{\text{s}}$, 
	$t_{\text{D,s}}  = \sum\nolimits_{k\in \mathcal{K}}{\left( {{T}_{\text{p}}}{{P}_{\text{p},k}}+{{T}_{\text{d}}}{{P}_{\text{d},k}} \right){{\mathbf{u}}_{\text{s}}^{\mathsf H}}{{\mathbf{R}}_{\text{err},k}}\mathbf{u}_{\text{s}}}+T\left\| {{\mathbf{u}}_{\text{s}}} \right\|_{2}^{2}\sigma^2$, and 
	${{\mathbf{R}}_{\text{err},k}} = {\mathbf{R}_k} - {{\mathbf{R}}_{\text{est},k}}=\sigma^2{{\mathbf{R}}_{k}}{\left( {{P}_{\text{p},k}}{{T}_{\text{p}}}{{\mathbf{R}}_{k}}+\sigma^2\mathbf{I} \right)^{-1}}$.
	Then, $\nabla_{P_{\mathrm{p},\ell}} t_{\text{N,s}}$ can be derived as 
	\begin{align}
		\label{eq:grad_t_Ns}
		&\nabla_{P_{\mathrm{p},\ell}} t_{\text{N,s}} 
		=  \left(2T_{\mathrm{p}}^{2}P_{\mathrm{p},\ell}+T_{\mathrm{d}}T_{\mathrm{p}}P_{\mathrm{d},\ell}\right) {{\mathbf{u}}_{\text{s}}^{\mathsf H}}{{\mathbf{R}}_{\mathrm{g},\ell}} \mathbf{C}_\ell^{-1}{{\mathbf{R}}_{\mathrm{g},\ell}}\mathbf{u}_{\text{s}}
		\notag \\ &
		-\left( {{T}_{\text{p}}^3}{{P}_{\text{p},\ell}^2}+{{T}_{\text{d}}}{{T}_{\text{p}}^2}{{P}_{\text{d},\ell}{{P}_{\text{p},\ell}}} \right) {{\mathbf{u}}_{\text{s}}^{\mathsf H}}{{\mathbf{R}}_{\mathrm{g},\ell}} \mathbf{C}_\ell^{-1} {{\mathbf{R}}_{\ell}} \mathbf{C}_\ell^{-1} {{\mathbf{R}}_{\mathrm{g},\ell}}\mathbf{u}_{\text{s}},
	\end{align}
	while $\ell=k$. And if $\ell \neq k$, $\nabla_{P_{\mathrm{p},\ell}} t_{\text{N,s}} = 0$.
	Similarly, $\nabla_{P_{\mathrm{p},\ell}} t_{\text{D,s}}$ can be obtained as
	\begin{align}
		\label{eq:grad_t_Ds}
		\nabla_{P_{\mathrm{p},\ell}} t_{\text{D,s}}
		=&-\sigma^2 \left( {{T}_{\text{p}}^2}{{P}_{\text{p},\ell}}+{{T}_{\text{p}}}{{T}_{\text{d}}}{{P}_{\text{d},\ell}} \right){{\mathbf{u}}_{\text{s}}^{\mathsf H}}{\mathbf{R}}_\ell {\mathbf{C}}_\ell^{-1}{\mathbf{R}}_\ell {\mathbf{C}}_\ell^{-1} {{\mathbf{u}}_{\text{s}}}
		\notag \\ &
		+\sigma^2 {{T}_{\text{p}}} {{\mathbf{u}}_{\text{s}}^{\mathsf H}} {\mathbf{R}}_\ell {\mathbf{C}}_\ell^{-1}{{\mathbf{u}}_{\text{s}}},
	\end{align}
	while $\ell=k$. And if $\ell \neq k$, $\nabla_{P_{\mathrm{p},\ell}} t_{\text{D,s}} = 0$.
	For the sake of simplicity, we define $\acute{t_{\text{N,s}}} =\nabla_{P_{\mathrm{p},\ell}} t_{\text{N,s}} $ and $\acute{t_{\text{D,s}}} =\nabla_{P_{\mathrm{p},\ell}} t_{\text{D,s}} $.
	Replacing ~\eqref{eq:grad_t_Ns} and~\eqref{eq:grad_t_Ds} in~\eqref{eq:grad_f_s},  $\nabla_{P_{\mathrm{p},\ell}}f_{\mathrm{s}}(\mathbf{p}_{\mathrm{p}},\tau_{\mathrm{s}})$ can be obtained. 
	Finally, the closed-form expression for $\nabla_{P_{\mathrm{p},\ell}}\mathcal{L}_{\bm{\upomega},\xi}(\mathbf{p}_{\text{p}},\bm{\uptau})$ in Theorem~\ref{thm:grad_Pd} is derived; this concludes the proof. 
\end{appendices}

% \balance 

% \input{references}

\bibliographystyle{IEEEtran}
\bibliography{references}

@ARTICLE{23-COMST-6G,
	author={Chafii, Marwa and Bariah, Lina and Muhaidat, Sami and Debbah, Merouane},
	journal={IEEE Commun. Surveys Tuts.}, 
	title={Twelve Scientific Challenges for 6{G}: {R}ethinking the Foundations of Communications Theory}, 
	year={2nd Quart., 2023},
	volume={25},
	number={2},
	pages={868-904}}

@ARTICLE{21-NETW-ISAC,
  author={Cui, Yuanhao and Liu, Fan and Jing, Xiaojun and Mu, Junsheng},
  journal={IEEE Netw.}, 
  title={Integrating Sensing and Communications for Ubiquitous {I}o{T}: {A}pplications, Trends, and Challenges}, 
  year={Sep./Oct. 2021},
  volume={35},
  number={5},
  pages={158-167},
  doi={10.1109/MNET.010.2100152}}

@ARTICLE{25-OJCOMS-ISAC-RIS,
  author={Magbool, Ahmed and Kumar, Vaibhav and Wu, Qingqing and Di Renzo, Marco and Flanagan, Mark F.},
  journal={IEEE Open J. Commun. Soc.}, 
  title={A Survey on Integrated Sensing and Communication With Intelligent Metasurfaces: {T}rends, Challenges, and Opportunities}, 
  year={2025},
  volume={6},
  number={},
  pages={7270-7318},
  doi={10.1109/OJCOMS.2025.3594049}}

@ARTICLE{20-NETW-6G,
  author={Saad, Walid and Bennis, Mehdi and Chen, Mingzhe},
  journal={IEEE Netw.}, 
  title={A Vision of 6{G} Wireless Systems: {A}pplications, Trends, Technologies, and Open Research Problems}, 
  year={May/Jun. 2020},
  volume={34},
  number={3},
  pages={134-142},
  doi={10.1109/MNET.001.1900287}}

@ARTICLE{22-COMST-ISAC,
  author={Liu, An and others},
  journal={IEEE Commun. Surveys Tuts.}, 
  title={A Survey on Fundamental Limits of Integrated Sensing and Communication}, 
  year={2nd Quart., 2022},
  volume={24},
  number={2},
  pages={994-1034},
  doi={10.1109/COMST.2022.3149272}}

@ARTICLE{23-IoTJ-ISAC,
  author={Wei, Zhiqing and others},
  journal={IEEE Internet Things J.}, 
  title={Integrated Sensing and Communication Signals Toward 5{G-A} and 6{G}: {A} Survey}, 
  year={Jul.2023},
  volume={10},
  number={13},
  pages={11068-11092},
  doi={10.1109/JIOT.2023.3235618}}

@ARTICLE{19-TVT-ISAC,
  author={Zhang, J. Andrew and Huang, Xiaojing and Guo, Y. Jay and Yuan, Jinhong and Heath, Robert W.},
  journal={IEEE Trans. Veh. Technol.}, 
  title={Multibeam for Joint Communication and Radar Sensing Using Steerable Analog Antenna Arrays}, 
  year={Jan. 2019},
  volume={68},
  number={1},
  pages={671-685},
  doi={10.1109/TVT.2018.2883796}}

@ARTICLE{22-JSAC-ISAC,
  author={Liu, Fan and others},
  journal={IEEE J. Sel. Areas Commun.}, 
  title={Integrated Sensing and Communications: {T}oward Dual-Functional Wireless Networks for 6{G} and Beyond}, 
  year={Jun. 2022},
  volume={40},
  number={6},
  pages={1728-1767}}

@ARTICLE{21-JSTSP-ISAC,
  author={Zhang, J. Andrew and others},
  journal={IEEE J. Sel. Top. Signal Process.}, 
  title={An Overview of Signal Processing Techniques for Joint Communication and Radar Sensing}, 
  year={Nov. 2021},
  volume={15},
  number={6},
  pages={1295-1315}}

@ARTICLE{21-VTM-ISAC-PMN,
  author={Zhang, Andrew and Rahman, Md. Lushanur and Huang, Xiaojing and Guo, Yingjie Jay and Chen, Shanzhi and Heath, Robert W.},
  journal={IEEE Veh. Technol. Mag.}, 
  title={Perceptive Mobile Networks: {C}ellular Networks With Radio Vision via Joint Communication and Radar Sensing}, 
  year={Jun. 2021},
  volume={16},
  number={2},
  pages={20-30},
  doi={10.1109/MVT.2020.3037430}}

@ARTICLE{22-IDAC-UL,
  author={Ouyang, Chongjun and Liu, Yuanwei and Yang, Hongwen},
  journal={IEEE Commun. Lett.}, 
  title={On the Performance of Uplink {ISAC} Systems}, 
  year={Aug. 2022},
  volume={26},
  number={8},
  pages={1769-1773},
  doi={10.1109/LCOMM.2022.3178193}}

@ARTICLE{11-PI-ISAC-Wave,
  author={Sturm, Christian and Wiesbeck, Werner},
  journal={Proc. IEEE}, 
  title={Waveform Design and Signal Processing Aspects for Fusion of Wireless Communications and Radar Sensing}, 
  year={Jul. 2011},
  volume={99},
  number={7},
  pages={1236-1259},
  doi={10.1109/JPROC.2011.2131110}}

@ARTICLE{20-TCOM-ISAC,
  author={Liu, Fan and Masouros, Christos and Petropulu, Athina P. and Griffiths, Hugh and Hanzo, Lajos},
  journal={IEEE Trans. Commun.}, 
  title={Joint Radar and Communication Design: {A}pplications, State-of-the-Art, and the Road Ahead}, 
  year={Jun. 2020},
  volume={68},
  number={6},
  pages={3834-3862},
  doi={10.1109/TCOMM.2020.2973976}}

@ARTICLE{21-JSTSP-ISAC-OTFS,
  author={Yuan, Weijie and Wei, Zhiqiang and Li, Shuangyang and Yuan, Jinhong and Ng, Derrick Wing Kwan},
  journal={IEEE J. Sel. Topics Signal Process.}, 
  title={Integrated Sensing and Communication-Assisted Orthogonal Time Frequency Space Transmission for Vehicular Networks}, 
  year={Nov. 2021},
  volume={15},
  number={6},
  pages={1515-1528},
  doi={10.1109/JSTSP.2021.3117404}}

@ARTICLE{22-COMST-ISAC-MN,
  author={Zhang, J. Andrew and Rahman, Md. Lushanur and Wu, Kai and Huang, Xiaojing and Guo, Y. Jay and Chen, Shanzhi and Yuan, Jinhong},
  journal={IEEE Commun. Surv. Tut.}, 
  title={Enabling Joint Communication and Radar Sensing in Mobile Networks—{A} Survey}, 
  year={1st Quart., 2022},
  volume={24},
  number={1},
  pages={306-345},
  doi={10.1109/COMST.2021.3122519}}

@ARTICLE{22-TVT-ISAC-RIS,
  author={Wang, Xinyi and Fei, Zesong and Huang, Jingxuan and Yu, Hanxiao},
  journal={IEEE Trans. Veh. Technol.}, 
  title={Joint Waveform and Discrete Phase Shift Design for {RIS}-Assisted Integrated Sensing and Communication System Under {C}ramer-{R}ao Bound Constraint}, 
  year={Jan. 2022},
  volume={71},
  number={1},
  pages={1004-1009},
  doi={10.1109/TVT.2021.3122889}}

@ARTICLE{23-TWC-ISAC-UAV,
  author={Meng, Kaitao and Wu, Qingqing and Ma, Shaodan and Chen, Wen and Wang, Kunlun and Li, Jun},
  journal={IEEE Trans. Wireless Commun.}, 
  title={Throughput Maximization for {UAV}-Enabled Integrated Periodic Sensing and Communication}, 
  year={Jan. 2023},
  volume={22},
  number={1},
  pages={671-687},
  doi={10.1109/TWC.2022.3197623}}

@ARTICLE{23-TVT-ISAC,
  author={Hua, Haocheng and Xu, Jie and Han, Tony Xiao},
  journal={IEEE Trans. Veh. Technol.}, 
  title={Optimal Transmit Beamforming for Integrated Sensing and Communication}, 
  year={Aug. 2023},
  volume={72},
  number={8},
  pages={10588-10603},
  doi={10.1109/TVT.2023.3262513}}

@ARTICLE{18-TSP-ISAC,
  author={Liu, Fan and Zhou, Longfei and Masouros, Christos and Li, Ang and Luo, Wu and Petropulu, Athina},
  journal={IEEE Trans. Signal Process.}, 
  title={Toward Dual-functional Radar-Communication Systems: {O}ptimal Waveform Design}, 
  year={Aug. 2018},
  volume={66},
  number={16},
  pages={4264-4279},
  doi={10.1109/TSP.2018.2847648}}

@ARTICLE{22-JSAC-ISAC-FD,
  author={Xiao, Zhiqiang and Zeng, Yong},
  journal={IEEE J. Sel. Areas Commun.}, 
  title={Waveform Design and Performance Analysis for Full-Duplex Integrated Sensing and Communication}, 
  year={Jun. 2022},
  volume={40},
  number={6},
  pages={1823-1837},
  doi={10.1109/JSAC.2022.3155509}}

@ARTICLE{23-TWC-ISAC,
  author={An, Jiancheng and Li, Hongbin and Ng, Derrick Wing Kwan and Yuen, Chau},
  journal={IEEE Trans. Wireless Commun.}, 
  title={Fundamental Detection Probability vs. Achievable Rate Tradeoff in Integrated Sensing and Communication Systems}, 
  year={Dec. 2023},
  volume={22},
  number={12},
  pages={9835-9853},
  doi={10.1109/TWC.2023.3273850}}

@ARTICLE{24-TWC-ISAC-Pilot,
  author={Hua, Meng and Wu, Qingqing and Chen, Wen and Jamalipour, Abbas and Wu, Celimuge and Dobre, Octavia A.},
  journal={IEEE Trans. Wireless Commun.}, 
  title={Integrated Sensing and Communication: {J}oint Pilot and Transmission Design}, 
  year={Nov. 2024},
  volume={23},
  number={11},
  pages={16017-16032},
  doi={10.1109/TWC.2024.3435912}}

@ARTICLE{22-TSP-ISAC-RIS,
  author={Yu, Zhouyuan and Hu, Xiaoling and Liu, Chenxi and Peng, Mugen and Zhong, Caijun},
  journal={IEEE Trans. Signal Process.}, 
  title={Location Sensing and Beamforming Design for {IRS}-Enabled Multi-User {ISAC} Systems}, 
  year={2022},
  volume={70},
  number={},
  pages={5178-5193},
  doi={10.1109/TSP.2022.3217353}}

@ARTICLE{24-TVT-ISAC,
  author={Ni, Zhitong and Zhang, J. Andrew and Huang, Xiaojing and Liu, Ren Ping},
  journal={IEEE Trans. Veh. Technol.}, 
  title={Frequency-Time Resource Allocation for Multiuser Uplink {ISAC} Systems}, 
  year={Dec. 2024},
  volume={73},
  number={12},
  pages={18893-18906},
  doi={10.1109/TVT.2024.3437683}}

@ARTICLE{25-TCCN-ISAC-UL,
  author={Li, Yiheng and Wei, Zhiqing and Liu, Haotian and Feng, Zhiyong},
  journal={IEEE Trans. Cogn. Commun. Netw.}, 
  title={Uplink Multi-{RSU} Cooperative Sensing Strategy for Integrated Sensing and Communication System}, 
  year={Oct. 2025},
  volume={11},
  number={5},
  pages={3112-3127},
  doi={10.1109/TCCN.2025.3570449}}

@article{17-bjornson-massiveMIMO,
  title={Massive {MIMO} networks: {S}pectral, energy, and hardware efficiency},
  author={Bj{\"o}rnson, Emil and Hoydis, Jakob and Sanguinetti, Luca},
  journal={Found. Trends Signal Process.},
  volume={11},
  number={3-4},
  pages={154--655},
  year={2017},
  publisher={Now Publishers, Inc.}
}

@book{98-SP-Dectect,
  author    = {Steven M. Kay},
  title     = {Fundamentals of Statistical Signal Processing, Volume II: Detection Theory},
  publisher = {Prentice-Hall},
  address   = {Englewood Cliffs, NJ, USA},
  year      = {1998}
}

@INPROCEEDINGS{25-ICC-ISCPT,
  author={Kumar, Vaibhav and Chafii, Marwa},
  booktitle={Proc. IEEE Int. Conf. Commun. (ICC)}, 
  title={Green Integration of Sensing, Communication, and Power Transfer via {STAR-RIS}}, 
  year={Jun. 2025},
  address={Montreal, Canada},
  volume={},
  number={},
  pages={2894-2899},
  doi={10.1109/ICC52391.2025.11161193}}

@article{66-Mini,
  title={Minimization of functions having {L}ipschitz continuous first partial derivatives},
  author={Armijo, Larry},
  journal={Pacific J. Math.},
  volume={16},
  number={1},
  pages={1--3},
  year={1966},
  publisher={Mathematical Sciences Publishers}
}

@INPROCEEDINGS{25-WCNC-RIS-MISO,
  author={Bahingayi, Eduard E. and Perović, Nemanja Stefan and Tran, Le-Nam},
  booktitle={Proc. IEEE Wireless Commun. Netw. Conf. (WCNC)}, 
  title={Weighted Sum-Rate Maximization for Large-Scale {RIS}-Assisted Multi-User {MISO} Systems}, 
  year={Mar. 2025},
  address={Milan, Italy},
  volume={},
  number={},
  pages={1-6},
  doi={10.1109/WCNC61545.2025.10978181}}

@ARTICLE{23-WCL-RIS,
  author={Kumar, Vaibhav and Zhang, Rui and Renzo, Marco Di and Tran, Le-Nam},
  journal={IEEE Wireless Commun. Lett.}, 
  title={A Novel {SCA}-Based Method for Beamforming Optimization in {IRS/RIS}-Assisted {MU-MISO} Downlink}, 
  year={Feb. 2023},
  volume={12},
  number={2},
  pages={297-301},
  doi={10.1109/LWC.2022.3224316}}

\end{document}